\documentclass[12pt, reqno]{amsart}
\makeatletter
\usepackage[margin=2cm]{geometry}
\usepackage{amsmath, amssymb, amsfonts, amstext, verbatim, amsthm, mathrsfs}
\usepackage[mathcal]{eucal}
\usepackage[all]{xy}
\let\oldtocsection=\tocsection

\let\oldtocsubsection=\tocsubsection

\let\oldtocsubsubsection=\tocsubsubsection

\renewcommand{\tocsection}[2]{\hspace{0em}\oldtocsection{#1}{#2}}
\renewcommand{\tocsubsection}[2]{\hspace{1em}\oldtocsubsection{#1}{#2}}
\renewcommand{\tocsubsubsection}[2]{\hspace{2em}\oldtocsubsubsection{#1}{#2}}

\usepackage[scr=boondox]{mathalpha}

\usepackage[usenames]{color}
\usepackage{aliascnt} 
\usepackage[colorlinks=true,linkcolor=blue,citecolor=blue,urlcolor=blue,citebordercolor={0 0 1},urlbordercolor={0 0 1},linkbordercolor={0 0 1}]{hyperref} 
\usepackage{enumerate}
\usepackage{tikz}
\usetikzlibrary{backgrounds}
\theoremstyle{plain}

\usepackage[numbers,sort&compress]{natbib}
\numberwithin{equation}{section}
\usepackage{appendix}

\newcommand{\refnewtheoremn}[4]{
\newaliascnt{#1}{#2}
\newtheorem{#1}[#1]{#3}
\aliascntresetthe{#1}
\expandafter\providecommand\csname #1autorefname\endcsname{#4}}

\newcommand{\refnewtheorem}[3]{\refnewtheoremn{#1}{#2}{#3}{#3}}

\def\makeCal#1{
\expandafter\newcommand\csname c#1\endcsname{\mathcal{#1}}}
\def\makeBB#1{
\expandafter\newcommand\csname b#1\endcsname{\mathbb{#1}}}
\def\makeFrak#1{
\expandafter\newcommand\csname f#1\endcsname{\mathfrak{#1}}}

\count@=0
\loop
\advance\count@ 1
\edef\y{\@Alph\count@}
\expandafter\makeCal\y
\expandafter\makeBB\y
\expandafter\makeFrak\y
\ifnum\count@<26
\repeat

\theoremstyle{definition}
\newtheorem{thm}{Theorem}[section]

\refnewtheorem{prethm}{thm}{Pre-theorem}

\refnewtheorem{lemma}{thm}{Lemma}
\refnewtheorem{claim}{thm}{Claim}
\refnewtheorem{cor}{thm}{Corollary}
\refnewtheorem{conj}{thm}{Conjecture}
\refnewtheorem{prop}{thm}{Proposition}
\refnewtheorem{quest}{thm}{Question}
\refnewtheorem{amplif}{thm}{Amplification}
\refnewtheorem{scholium}{thm}{Scholium}

\refnewtheorem{tablethm}{thm}{Table}
\refnewtheorem{assumptions}{thm}{Assumptions}
\refnewtheorem{conclusion}{thm}{Conclusion}
\refnewtheorem{problem}{thm}{Problem}
\refnewtheorem{ansatz}{thm}{Ansatz}

\theoremstyle{definition}
\refnewtheorem{remark}{thm}{Remark}
\refnewtheorem{remarks}{thm}{Remarks}
\refnewtheorem{defn}{thm}{Definition}
\refnewtheorem{notn}{thm}{Notation}
\refnewtheorem{const}{thm}{Construction}
\refnewtheorem{ex}{thm}{Example}
\refnewtheorem{fact}{thm}{Fact}
\refnewtheorem{recollection}{thm}{Recollection}

\renewcommand{\Im}{\operatorname{Im}}
\renewcommand{\Re}{\operatorname{Re}}
\newcommand {\Hom}{\operatorname{Hom}}

\newcommand{\ch}{\operatorname{ch}}
\newcommand{\Coh}{\operatorname{Coh}}

\newcommand{\SL}{\operatorname{SL}}

\newcommand{\GL}{\operatorname{GL}}

\renewcommand{\dim}{\operatorname{dim}}

\newcommand{\Li}{\operatorname{Li}}

\newcommand {\<}{\langle}
\renewcommand {\>}{\rangle}

\newcommand{\half}{\tfrac{1}{2}}

\newcommand{\D}{D}

\renewcommand{\b}{\,|\,}
\newcommand{\gammac}{\gamma^\vee}

\newcommand{\mat}[4]{\begin{pmatrix}#1&#2\\#3&#4\end{pmatrix}}
\newcommand{\oomit}[1]{}

\usepackage{calligra}
\DeclareMathAlphabet{\mathcalligra}{T1}{calligra}{m}{n}
\DeclareFontShape{T1}{calligra}{m}{n}{<->s*[2.2]callig15}{}

\begin{document}

\title{ Resurgence and Riemann-Hilbert problems for  elliptic Calabi-Yau threefolds}
\author{Tom Bridgeland and Iv{\'a}n Tulli}
\date{}

\begin{abstract}{Let $X$ be a  Calabi-Yau threefold with an elliptic fibration. We investigate the non-linear Riemann-Hilbert  problems associated to the Donaldson-Thomas  theory of fibre classes, and relate them to the Borel sum of the $A$-model topological string free energy for such classes. } \end{abstract}

\maketitle


\section{Introduction}\label{introsec}

There has been a great deal of interest  recently in applying techniques from resurgence analysis to study non-perturbative effects in topological string theory.  A central object in the theory is the free energy, which is a formal series in the topological string coupling $\lambda$. More precisely, the  A-model closed string free energy of a Calabi-Yau (CY) threefold $X$
(in the holomorphic limit)
is a series of the form \begin{equation}F(\lambda,Q)=\sum_{g\geq 0}F_g(Q)\, \lambda^{2g-2},\end{equation} where $F_g(Q)$ coincides with the generating function for genus $g$ Gromov-Witten (GW) invariants of $X$.  In studying non-perturbative effects it has been very fruitful to consider, where possible, the Borel sum of the free energy \cite{ASTT, IM, GM, GKKM, G}.
The latter depends on a choice of a ray $r\subset \mathbb{C}^*$, with the Borel sum for different  choices being related by Stokes jumps. These jumps are conjecturally related to the Donaldson-Thomas (DT) invariants of $X$. 

One (non-compact) CY threefold  that has been much-studied from this point of view is  the resolved conifold. The Borel summability of its A-model free energy was established in \cite{PS,HO,ASTT}. The resulting non-perturbative free energy is closely related to the log of the triple sine function $\sin_3(z\b \omega_1,\omega_2,\omega_3)$\footnote{The multiple sine functions are in turn defined in terms of the multiple gamma functions of Barnes \cite{Bar}.} \cite{KK}. The Stokes behaviour of the Borel sums was completely described in \cite{ASTT} and shown to be controlled by the corresponding DT invariants. 

These ideas make contact with a closely related area of research which aims to use the DT invariants of a CY$_3$ category to construct a geometric structure on the space of stability conditions \cite{RHDT2,BS,B2,BM, Btau}. The  geometric structure goes by the name of a Joyce structure, and  is built from solutions to a collection of non-linear Riemann-Hilbert (RH) problems. These problems  involve piecewise holomorphic functions on  $\bC^*$, with discontinuities prescribed by   the DT invariants along a collection of rays, and fixed asymptotics at $0$ and $\infty$. The existence and uniqueness of solutions for such problems is not known in general, but several examples have been worked out in detail   \cite{RHDT2,BM,B}.

In the case of the resolved conifold the RH problems  have unique solutions \cite{B,AP}, which are closely related to the Borel sums of the free energy. More precisely, this statement holds after appropriately fixing the  constant term in the asymptotics at $\epsilon=0$. The solutions are then given by functions $Y_i^{r}(\epsilon \b \omega_1,\omega_2)$, with $i=1,2$, where $r=\bR_{>0}\cdot \zeta\subset \mathbb{C}^*$ is a ray, and $\epsilon\in \bH_r$ lies in the open half-plane $\mathbb{H}_r\subset \bC$ centered on $r$. They can be repackaged in terms of functions $\tau^r(\epsilon \b \omega_1,\omega_2)$ satisfying the equations
\begin{equation}\label{taudef}
       \frac{\partial}{\partial \omega_i} \log\tau^r(\epsilon\b \omega_1,\omega_2)=\frac{1}{2\pi \mathrm{i}}\cdot \frac{\partial}{\partial \epsilon}  \log Y_{i}^r(\epsilon\b\omega_1,\omega_2).
\end{equation}
It is these functions $\tau^r(\epsilon \b \omega_1,\omega_2)$ which are closely  related to the Borel summation of the free energy along the ray $r$ \cite{B,A,ASTT}. 

The goal of this paper is to address similar questions for compact CY threefolds with elliptic fibrations. 
We  study the Borel sums of the  A-model free energy and its relation to the RH problem defined by the DT invariants. We only consider the part of the free energy of $X$ corresponding to fibre classes, i.e. classes $\beta\in H_2(X,\bZ)$ satisfying $\pi_*(\beta)=0$. Similarly we only consider DT invariants for coherent sheaves supported on the fibres of $\pi$. 

 The rest of the introduction contains a detailed summary of our main results. In general we find a similar situation to that of the resolved conifold, although our results are not as complete.  One significant additional difficulty  is that  whereas in the case of the resolved conifold the set of Stokes directions is a closed subset of the circle, for the Borel sums  considered in this paper the Stokes directions are everywhere dense. Nonetheless, we find that the free energy is again Borel summable at least along almost all non-Stokes rays, and we construct natural solutions to a weak version of the RH problem which ignores the asymptotics at $\infty$. Moreover, the two stories are again related by the equation \eqref{taudef}.
 
\subsection{Borel sum of the free energy}\label{introBorel}

In order to state our results in more detail let us briefly recall the basics of Borel summation.
 Consider a formal power series $F(\epsilon)=\sum_{i\geq 1} a_i \epsilon^i$. The Borel transform is the series $f(\eta)=\sum_{i\geq 1} a_i\eta^{i-1}/(i-1)!$. For simplicity let us assume that $f(\eta)$ is the Taylor expansion of a meromorphic function on $\bC$ which we also denote  by $f(\eta)$. We consider rays $r\subset \bC^*$ of the form $r=\bR_{>0}\cdot \zeta$ with $\zeta\in \bC^*$. 
 Such a ray  is called a Stokes ray if it contains a pole of $f(\eta)$, otherwise it is a non-Stokes ray.
For a given $\epsilon \in \bC^*$ the series $F(\epsilon)$ is said to be Borel summable along a non-Stokes ray $r\subset \bC^*$ if the integral \begin{equation}
    \label{transform}
F^r(\epsilon)=\int_r e^{-\eta/\epsilon} f(\eta)\mathrm{d}\eta\end{equation} exists.  The Borel sum is then defined to be the value of this integral. In practice, given a non-Stokes ray $r\subset \bC^*$, we shall only consider the Borel sum for $\epsilon$ lying in the open half-plane $\bH_r=\{z\in \bC:\Re(z/\zeta)>0\}$ centered on $r$.  

We will be interested in applying Borel summation to the fibre-class free energy of the A-model topological string on   an elliptic CY threefold $X$.

\begin{assumptions}
\label{CYass}
By an elliptic CY threefold we mean a smooth projective threefold $X$, with trivial canonical bundle, equipped with a flat map $\pi\colon X\to B$  whose general fibre is a genus 1 curve. We always assume  that $B$ is smooth, that the singular fibers of $\pi$ are reduced and irreducible, and that $\pi$ has a section.
We further assume that the DT/GW correspondence holds for $X$.\end{assumptions}

Under these assumptions the GW  invariants of $X$ in the fibre classes were  computed
in \cite[Section B.3]{OP}. See Appendix \ref{appendix} for a more detailed discussion. This leads to an  expression 
 \begin{equation}
 \label{free_energy}F_{\text{GW}}(\lambda\b \tau)=-e(X)\cdot \sum_{g\geq 2} \frac{B_{2g} \, G_{2g-2}(\tau)}{4g(2g-2)} \left(\frac{\lambda}{2\pi}\right)^{2g-2}\end{equation}
for the $g\geq 2$ part of the fibre-class free energy. It is a formal series in $\lambda$ whose coefficients depend on a K{\"a}hler parameter $\tau\in \bC$ satisfying $\Im(\tau)>0$. More precisely, $\tau$ is the pairing of the complexified K{\"a}hler class $B+\mathrm{i}\omega\in H^2(X,\bC)$ with the fundamental class $\beta\in H_2(X,\bZ)$  of a smooth fibre of $\pi$. The expression involves the Bernoulli numbers $B_{2g}$, the Eisenstein series $G_{2g-2}(\tau)$, and the topological Euler characteristic $e(X)$.

It will be convenient to set $2\pi \mathrm{i} \epsilon=\lambda/2\pi$ and view $F_{\text{GW}}(\lambda\b\tau)$ as a formal series in $\epsilon$. Furthermore, in order to relate the Borel summations of $F_{\text{GW}}(\epsilon\b \tau)$ to the RH problem below, we consider instead 
\begin{equation}\label{free_energyres}
    F_{\text{GW}}(\epsilon \b \omega_1,\omega_2):=F_{\text{GW}}(\epsilon/\omega_1\b \omega_2/\omega_1)\,,
\end{equation}
where $\epsilon,\omega_1,\omega_2\in \mathbb{C}^*$ and $\text{Im}(\omega_2/\omega_1)>0$.
All our results about  $F_{\text{GW}}(\epsilon \b \omega_1,\omega_2)$ then specialize to results about \eqref{free_energy} by simply taking $\omega_1=1$ and $\omega_2=\tau$.
Our main result concerning the Borel summation of $F_{\text{GW}}(\epsilon \b \omega_1,\omega_2)$  is as follows:
  
 \begin{thm}
 \label{one}
Fix $\omega_1,\omega_2\in \mathbb{C}^{*}$ with $\Im(\omega_2/\omega_1)>0$. \begin{itemize}
    \item[(i)]
The Borel transform of  $F_{\text{GW}}(\epsilon\b\omega_1,\omega_2)$ is the Taylor expansion of a meromorphic function on $\bC$ with double poles at the points $a_1\omega_1+a_2\omega_2$ with $(a_1,a_2)\in \bZ^2\setminus\{0\}$ and no other poles.  \smallskip
\item[(ii)]For almost all non-Stokes rays $r\subset \bC^*$ the series $F_{\text{GW}}(\epsilon\b\omega_1,\omega_2)$ is Borel summable along $r$ provided $\epsilon$  lies in the  corresponding open half-plane $\bH_r\subset \bC^*$.\qed \end{itemize} 
\end{thm} 

More precisely, given a non-Stokes ray $r\subset \bC^*$, there is a unique real number $\alpha \in \mathbb{R}\setminus \mathbb{Q}$ such that $\pm (\omega_1+\alpha \omega_2) \in r$. We show that the Borel sum $F^{r}_{\text{GW}}(\epsilon\b \omega_1,\omega_2)$ exists and defines a holomorphic function of $\epsilon\in \bH_r$ whenever $\alpha$ does  not lie in the measure-zero subset of $\bR\setminus\bQ $ consisting of Liouville irrationals. For a general non-Stokes ray $r\subset \bC^*$  we can still associate a meaningful holomorphic function of $\epsilon \in \mathbb{H}_r$ by using integrals along certain detour paths (see Section \ref{Borelssummary}). These integrals reduce to the Borel sums from Theorem \ref{one} whenever the ray $r$  corresponds to  an element $\alpha\in \bR\setminus\bQ$ which is not a Liouville irrational.

\subsection{DT invariants and the RH problem}\label{DTRHintrosec}

Let $\pi\colon X\to B$ be an elliptic CY threefold satisfying Assumptions \ref{CYass}.  We consider the full triangulated subcategory $\cD(\pi)\subset \D^b\Coh(X)$ of the bounded derived category of coherent sheaves  consisting of objects  whose set-theoretic support is contained in a finite union of fibres of $\pi$. The Chern character defines a homomorphism
\begin{equation}
    \label{ch}\ch\colon K_0(\cD(\pi))\to N(\pi)\subset H^*(X,\bZ),
\end{equation}
 whose image $N(\pi)=\bZ\gamma_1\oplus \bZ\gamma_2$ is a free abelian group of rank 2.  It is convenient to choose the generators $\gamma_1,\gamma_2\in N(\pi)$
so that if $E$ is a rank $r$, degree $d$ vector bundle supported on a smooth fibre of $\pi$ then  $\ch(E)=-d\gamma_1+r\gamma_2$.

 Given a pair of complex numbers $\omega_1,\omega_2\in \bC^*$ with $\Im(\omega_2/\omega_1)>0$ there is a natural stability condition on the  category $\cD(\pi)$, uniquely defined up to the action of the even shifts,  whose central charge $Z\colon K_0(\cD(\pi))\to \bC$  is the composition of the Chern character \eqref{ch} with the map
 \begin{equation}\label{central_charge}
Z\colon N(\pi)\to \bC, \qquad Z(a_1\gamma_1+a_2\gamma_2)=a_1\omega_1+a_2\omega_2.\end{equation}
A calculation of Toda  \cite[Thm. 6.9]{T}  shows  that    the corresponding DT invariants are  \begin{equation}\label{dt}\Omega(a_1\gamma_1+a_2\gamma_2)=-e(X), \qquad (a_1,a_2)\in \bZ^2\setminus\{0\}.
 \end{equation}

In \cite{RHDT} it was explained how to associate a RH problem to the data of the lattice $N(\pi)$, the central charge \eqref{central_charge}, and the DT invariants \eqref{dt}. We will recall the  details of this construction in Section \ref{appsec} below. Here we will simply state the resulting RH problem and discuss its solution.
 A ray $r\subset \bC^*$ will be called a Stokes ray if it contains a point of the form $Z(\gamma)$ with  $0\neq \gamma\in N(\pi)$, otherwise $r$ will be called non-Stokes.  As before,  given a ray $r\subset \bC^*$, we denote by $\bH_r\subset \bC^*$ the open half-plane centered on it.

 \begin{problem}\label{RHprob}
 For each non-Stokes ray $r\subset \bC^*$ find holomorphic functions $Y^{r}_{i}\colon \bH_r\to \bC^*$ for $i=1,2$ such that the following statements hold.
 \begin{itemize}
 \item[(RH1)] If $\Delta\subset \mathbb{C}^*$ be a convex sector whose boundary consists of non-Stokes rays $r_1,r_2$ taken in clockwise order then
 \begin{equation}
Y^{r_2}_{i}(\epsilon)=Y^{r_1}_{i}(\epsilon)\cdot \prod_{\gamma=a_1\gamma_1+a_2\gamma_2\in Z^{-1}(\Delta)}\left(1-e^{-Z(\gamma)/\epsilon}\right)^{-a_i\cdot e(X)}\end{equation}
for $\epsilon\in \bH_{r_1}\cap \bH_{r_2}$ with $0<|\epsilon|\ll 1$. 

\item[(RH2)] As $\epsilon\to 0$ in any closed subsector of $\bH_r$ we have $Y^r_i(\epsilon)\to 1$.
 \item[(RH3)] There is an $N>0$ such that as $\epsilon\to \infty$ in $\bH_r$ there is a bound $|\epsilon|^{-N}< |Y^r_i(\epsilon)|<|\epsilon|^N$.
\end{itemize}
\end{problem}

It is easy to see that if this problem has a solution then it is unique. We shall instead consider what we call the weak RH problem 
in which we drop condition (RH3). The resulting solution is then unique up to simultaneous multiplication of the functions $Y^r_{i}$ for all rays $r\subset \bC^*$ by  an arbitrary  pair of holomorphic functions $P_i\colon \bC\to \bC^*$ satisfying $P_i(0)=1$. 

In order to motivate our solution of the weak RH problem, consider again the holomorphic functions $F_{\text{GW}}^{r}(\epsilon \b \omega_1,\omega_2)$ from Section \ref{introBorel} and define
\begin{equation}
    \tau_{\text{GW}}^r(\epsilon \b \omega_1,\omega_2):=\exp(F_{\text{GW}}^r(\epsilon \b \omega_1,\omega_2))\,.
\end{equation}
As before, for a general non-Stokes ray $r$ it is understood that the above expression is defined via detour paths.
Furthermore, let $H_{\text{GW}}(\epsilon\b \omega_1,\omega_2)$ be the formal series in $\epsilon$ without constant term satisfying 
\begin{equation}
    \frac{\partial}{\partial \epsilon}H_{\text{GW}}(\epsilon \b \omega_1,\omega_2)=F_{\text{GW}}(\epsilon \b \omega_1,\omega_2)\,.
\end{equation}
When looking for solutions of the RH problem related to $\tau_{\text{GW}}^{r}(\epsilon \b \omega_1,\omega_2)$ via \eqref{taudef}, it is then natural to consider Borel summations of $2\pi \mathrm{i}\cdot \frac{\partial}{\partial \omega_i} H_{\text{GW}}(\epsilon \b \omega_1,\omega_2)$. Our second main result is then as follows:

 \begin{thm}
 \label{two} Fix $\omega_1,\omega_2\in \mathbb{C}^*$ with $\text{Im}(\omega_2/\omega_1)>0$.  Then there exists a solution $Y^{r}_{i}(\epsilon\b\omega_1,\omega_2)$ of the weak RH problem such that $\log Y_i^{r}(\epsilon \b \omega_1,\omega_2)$ is the Borel sum of $2\pi \mathrm{i}\frac{\partial}{\partial \omega_i} H_{\text{GW}}(\epsilon \b \omega_1,\omega_2)$ along $r$ for almost all non-Stokes rays $r$. Furthermore, 
\begin{equation}\label{tauRHw}
    \frac{\partial}{\partial \omega_i} \log\tau^r_{\text{GW}}(\epsilon\b \omega_1,\omega_2)=\frac{1}{2\pi \mathrm{i}}\cdot \frac{\partial}{\partial \epsilon}  \log Y^{r}_{i}(\epsilon\b\omega_1,\omega_2)
\end{equation}
for all non-Stokes rays $r$. \qed
 \end{thm}

 \subsection{Further remarks}

Our results leave several natural challenges and questions for future research.  For almost all  rays $r\subset \bC^*$ our  solution  to the weak RH problem can be expressed  \eqref{hear}  as an integral
\begin{equation}\label{hu}
    Y_{i}^r(\epsilon)=\exp\left(-\frac{e(X)}{2\pi\mathrm{i}}\int_{r} \mathrm{Li}_1(e^{- \eta/\epsilon})\frac{\partial}{\partial \omega_i} \mathscr{h}(\eta\b\omega_1,\omega_2)\mathrm{d}\eta\right)\,,
\end{equation}
where $\mathscr{h}(\eta\b\omega_1,\omega_2)$ is closely related to the log of the Jacobi theta function, and is defined in terms of the  Weierstrass sigma function by the equation
\[\mathscr{h}(u\b\omega_1,\omega_2)=\log \sigma(u\b\omega_1,\omega_2)-\log(u)-\half G_2(\omega_1,\omega_2)u^2.\]
An obvious challenge is to  upgrade Theorem \ref{two} by constructing a solution to the full Riemann-Hilbert problem. This would involve understanding the behaviour of the integral \eqref{hu} in the limit $\epsilon\to \infty$.

In the case of the resolved conifold, the Borel sum of the free energy along a particular ray can be re-expressed \cite[Theorem 2.1]{ASTT}   in terms of the Barnes triple sine function. It is natural to ask whether the integral \eqref{hu} can also be re-expressed in some more convenient form, and whether it can be related to known special functions. 

A very interesting property of the solution to the RH problem  in the case of the resolved conifold  is an unexpected symmetry exchanging  the parameter $\epsilon \in \bC^*$ with the central charge parameter corresponding to the  class of a point. A possible relation to S-duality in string theory was discussed in \cite[Section 6]{ASTT}.  For the RH problem considered in this paper the analogous symmetry would  exchange the parammeters $\epsilon$ and $\omega_2$. Since the solutions to the RH problem already have an obvious  $\SL_2(\bZ)$ symmetry acting on the parameters $\omega_1,\omega_2$, this perhaps hints at a possible connection with modular forms for $\SL_3(\bZ)$.

\subsection{Structure of the paper}
In Section \ref{analysis} we introduce some modified Weierstrass elliptic functions and summarize the results from complex analysis that we will need.   In Section \ref{appsec} we apply the contents of Section \ref{analysis} to prove our main results,  Theorems \ref{one} and \ref{two}. The proofs of the results in Section \ref{analysis} can be found  in Sections \ref{Boreltranssec} and  \ref{proofssec}. Section \ref{Boreltranssec}  is concerned with the properties of the Borel transforms, whereas Section \ref{proofssec} deals with  the proof of Borel summability  and  the existence of closely related integrals along detour paths.

\subsection{Acknowledgements} The authors are grateful to Georg Oberdieck who suggested applying the techniques of \cite{ASTT,B} to the case of elliptic CY threefolds.


\section{Summary of the relevant complex analysis}
\label{analysis}

In this section we collect the precise statements of the results from complex analysis which will be applied to prove Theorems \ref{one} and \ref{two} in Section \ref{appsec}. The proofs for the results in this section can be found in Sections \ref{Boreltranssec} and \ref{proofssec}. 

\subsection{Elliptic functions}
\label{elliptic}
Define the region
\begin{equation}\label{reg}
    \cR=\{(\omega_1,\omega_2)\in (\bC^*)^2: \Im (\omega_2/\omega_1)>0\}.
\end{equation}
A point $(\omega_1,\omega_2)\in \cR$ defines a lattice  \begin{equation}\Lambda(\omega_1,\omega_2)=\bZ \omega_1+\bZ \omega_2\subset \bC.\end{equation}
We set $\Lambda^*(\omega_1,\omega_2)=\Lambda(\omega_1,\omega_2)\setminus\{0\}$.
For an even integer $n\geq 2$ we introduce the Eisenstein series
 \begin{equation}
     \label{eisenstein}
G_{n}(\omega_1,\omega_2)=\sum_{\omega\in \Lambda^*(\omega_1,\omega_2)} \frac{1}{\omega^n}=\sum_{0\neq (a_1,a_2)\in \bZ^2} \frac{1}{(a_1\omega_1+a_2\omega_2)^n}.\end{equation}
This series is absolutely convergent for $n>2$, while for $n=2$ it is only conditionally convergent. We define $G_2$ by the Eisenstein summation
 \begin{equation}
     G_2(\omega_1,\omega_2):=\sum_{a_1\in \mathbb{Z}\setminus\{0\}}\frac{1}{(a_1\omega_1)^2}+\sum_{a_2\in \mathbb{Z}\setminus\{0\}}\sum_{a_1\in \mathbb{Z}}\frac{1}{(a_1\omega_1+a_2\omega_2)^2}\,.
 \end{equation}
The resulting functions $G_n(\omega_1,\omega_2)$ are holomorphic on $\cR$ for all $n\geq 2$. These are related to the $G_{n}(\tau)$ appearing in \eqref{free_energy} via $G_{n}(\tau)=G_{n}(1,\tau)$.

We recall the Weierstrass elliptic functions. The functions $\wp(u\b\omega_1,\omega_2)$ and $\zeta(u\b\omega_1,\omega_2)$  are meromorphic functions of $u\in \bC$ with poles of order 2 and 1 respectively at the lattice points $\Lambda(\omega_1,\omega_2)$. The function $\sigma(u\b\omega_1,\omega_2)$ is an entire function of $u\in \bC$ with simple zeroes at the lattice points. There are relations
\begin{equation}
\zeta(u\b\omega_1,\omega_2)=\frac{\partial}{\partial u} \log \sigma(u\b \omega_1,\omega_2), \qquad \wp(u\b\omega_1,\omega_2)=-\frac{\partial^2}{\partial u^2} \log \sigma(u\b \omega_1,\omega_2),
\end{equation}
and a Laurent expansion at $u=0$
\begin{equation}
    \label{expan}
\log \sigma(u\b \omega_1,\omega_2)-\log(u)=-\sum_{g\geq 3} \frac{G_{2g-2}(\omega_1,\omega_2)}{2g-2} \cdot  u^{2g-2}.\end{equation}

We introduce minor modifications
\begin{equation}
    \label{mod_h}
\mathscr{h}(u\b\omega_1,\omega_2)=\log \sigma(u\b\omega_1,\omega_2)-\log(u)-\half G_2(\omega_1,\omega_2)u^2,\end{equation}
\begin{equation}\label{mod_zeta}\tilde{\zeta}(u\b\omega_1,\omega_2)= \frac{\partial}{\partial u} \mathscr{h}(u\b \omega_1,\omega_2)=\zeta(u\b\omega_1,\omega_2)-u^{-1}-G_2(\omega_1,\omega_2) \, u ,\end{equation}
\begin{equation}\label{mod_wp} \tilde{\wp}(u\b\omega_1,\omega_2)=-\frac{\partial ^2}{\partial u^2} \mathscr{h}(u\b\omega_1,\omega_2)=\wp(u\b\omega_1,\omega_2)-u^{-2}+G_2(\omega_1,\omega_2),\end{equation}
which are holomorphic near $u=0$.
Below we shall need the related functions
\begin{equation}\label{curlydef1}\mathscr{f}(u\b\omega_1,\omega_2)=2\tilde{\zeta}(u\b\omega_1,\omega_2)-u \tilde{\wp}(u\b\omega_1,\omega_2), \end{equation}
\begin{equation}\label{curlydef2}\mathscr{k}_i(u\b\omega_1,\omega_2)=\frac{\partial}{\partial \omega_i} \mathscr{h} (u\b\omega_1,\omega_2),\end{equation}
which have poles precisely at the nonzero lattice points $\Lambda^*(\omega_1,\omega_2)$. These  are  double poles in the case of $\mathscr{f}$ and simple poles in the case of $\mathscr{k}_i$. Later we will need the parity properties
\begin{equation}
    \label{evenodd}
    \mathscr{f}(-u\b\omega_1,\omega_2)=-\mathscr{f}(u\b\omega_1,\omega_2),\qquad \mathscr{k}_i(-u\b\omega_1,\omega_2)=\mathscr{k}_i(u\b\omega_1,\omega_2)
\end{equation}
which follow immediately from the expansion \eqref{expan}.

\subsection{Borel transforms}

Our starting point is the following formal power series in $\epsilon$
\begin{equation} \label{power1} H(\epsilon\b\omega_1,\omega_2)=\sum_{g\geq 2} \frac{B_{2g} \,G_{2g-2}(\omega_1,\omega_2)\, (2\pi \mathrm{i})^{2g}\, \epsilon^{2g-1}}{4g(2g-1)(2g-2)}.\end{equation}
The coefficients are holomorphic functions of $(\omega_1,\omega_2)\in \cR$ involving the Bernoulli numbers $B_{2g}$ and the Eisenstein series \eqref{eisenstein}. We then  consider 
the formal power series
\begin{equation}\label{power2} F(\epsilon\b\omega_1,\omega_2)=\frac{\partial }{\partial \epsilon} H(\epsilon\b\omega_1,\omega_2),\qquad K_i(\epsilon\b\omega_1,\omega_2)=\frac{\partial }{\partial \omega_i} H(\epsilon\b\omega_1,\omega_2).\end{equation}
We denote by $f(\eta\b\omega_1,\omega_2)$ and $k_i(\eta\b\omega_1,\omega_2)$ the Borel transforms of these series. They are power series in $\eta$ with coefficients which are holomorphic functions of $(\omega_1,\omega_2)\in \cR$.

Note that $F(\epsilon\b\omega_1,\omega_2)$ is related to the previous $F_{\text{GW}}(\epsilon\b \omega_1,\omega_2)$ from \eqref{free_energyres}  by
\begin{equation}\label{FGWFrelres}
    F_{\text{GW}}(\epsilon \b \omega_1,\omega_2)=-\frac{e(X)}{(2\pi \mathrm{i})^2}\cdot F(\epsilon \b \omega_1, \omega_2)\,.
\end{equation}
In particular, for  $F_{\text{GW}}(\lambda \b \tau)$ given in \eqref{free_energy}
\begin{equation}\label{FGWFrel}
    F_{\text{GW}}(\lambda \b \tau)=-\frac{e(X)}{(2\pi \mathrm{i})^2}\cdot F(\epsilon \b 1, \tau), \qquad 2\pi \mathrm{i}\epsilon=\frac{\lambda}{2\pi}\,.
\end{equation}
We choose to work with $F(\epsilon \b \omega_1,\omega_2)$ rather than directly  with $F_{\text{GW}}(\lambda \b \tau)$ for two reasons. On the one hand, the change of variables from $\lambda$ to $\epsilon$ and the rescaling by $-e(X)/(2\pi \mathrm{i})^2$ eliminates certain awkward factors from the Borel sums and the positions of the poles of the Borel transform. On the other hand, the introduction of the variables $\omega_1,\omega_2$  facilitates the relation with the RH problem. 

Let us fix a point $(\omega_1,\omega_2)\in \cR$. The following result is proved in Section \ref{Boreltranssec}.

\begin{prop}
\label{propone} 
\begin{itemize}
    \item[(i)]
The Borel transforms   $f(\eta\b\omega_1,\omega_2)$ and $k_i(\eta\b\omega_1,\omega_2)$  have positive radius of convergence and hence define holomorphic functions in a neighbourhood of $\eta=0$.
\smallskip

\item[(ii)] These functions extend to meromorphic functions of $\eta\in \bC$ with poles precisely at the nonzero lattice points $\Lambda^*(\omega_1,\omega_2)$. The poles  are double poles in the case of $f$ and simple poles in the case of $k_i$.
\smallskip

\item[(iii)] 
  There are explicit expressions
\begin{equation}
    \label{curlyf}
f(\eta\b\omega_1,\omega_2)= \sum_{m\geq 1} \frac{1}{m^3} \, \mathscr{f}\Big(\frac{\eta}{m}\b \omega_1,\omega_2\Big), \end{equation}
\begin{equation}
    \label{curlyk2}
k_i(\eta\b\omega_1,\omega_2)= \sum_{m\geq 1} \frac{1}{m^2} \, \mathscr{k}_i\Big(\frac{\eta}{m}\b \omega_1,\omega_2\Big), \end{equation}
which converge absolutely and uniformly for $\eta$ in compact subsets of $\bC$. \qed
\end{itemize}

\end{prop}

The Borel transform of $H(\epsilon\b\omega_1,\omega_2)$ also has positive radius of convergence, but the  holomorphic continuation of the resulting function $h(\eta\b\omega_1,\omega_2)$ is more complicated due to  the presence of logarithmic singularities, and we will not directly consider this function here. 

\subsection{Irrationality measure}
\label{irrat}

 To define the Borel sum of the series  \eqref{power2} we must consider a Laplace-type integral of the form \eqref{transform}. Note that a ray $r\subset \bC^*$ is non-Stokes precisely if it contains no points of the lattice $\Lambda(\omega_1,\omega_2)\subset \bC$. 
 Since a non-Stokes ray still comes arbitrarily close to points of $\Lambda(\omega_1,\omega_2)$, when trying to control the growth of such integrals we encounter some basic notions from Diophantine approximation 
which we now recall.

The irrationality measure $\mu(\alpha)$ of a real number  $\alpha \in \mathbb{R}$ \cite[Definition E.1]{YB} is defined to be the infimum $\mu(\alpha)=\inf  R(\alpha)$ of the set
    \begin{equation}
        R(\alpha)=\big\{d \in \mathbb{R}_{>0} \; | \; 0<\left|\alpha - p/q\right|<1/q^d  \text{ for at most finitely many }p,q\in \mathbb{Z}, q>0\big\}.
    \end{equation}
    If $R(\alpha)= \emptyset$ we set $\mu(\alpha)=\infty$. In this case $\alpha$ is  known as a Liouville irrational.
    We will use the following well-known properties of $\mu(\alpha)$. 

\begin{thm}
\label{roth}

\begin{itemize}
    \item[(i)] if $\alpha\in \bQ$ then $\mu(\alpha)=1$,\smallskip
    
    \item[(ii)] if $\alpha \in \mathbb{R}\setminus\mathbb{Q}$ then $\mu(\alpha)\geq 2$, \smallskip
    
    \item[(iii)] if $\alpha\in \bR\setminus\bQ$ then
    \begin{equation} 
   \mu\bigg(\frac{a\alpha+b}{c\alpha+d}\bigg)=\mu(\alpha)\text{ for all } \mat{a}{b}{c}{d}\in \GL_2(\bZ),
   \end{equation}
   
   \item[(iv)] the subset $\{\alpha\in \bR:\mu(\alpha)>2\}$ has measure zero.
\end{itemize}
\end{thm}

\begin{proof}
If $\alpha \in \mathbb{Q}$ then it is easy to check that $\mu(\alpha)\geq 1$, while $\mu(\alpha)\leq 1$ follows from Liouville's theorem, which states that algebraic numbers of degree $n$ satisfy $\mu(\alpha)\leq n$. Part (ii) follows immediately from the Dirichlet approximation theorem, while (iv) is a Theorem due to Khinchin \cite{Kin}, whose proof is essentially an application of the Borel-Cantelli Lemma.  We could not find a direct reference for part (iii) so we include a proof in Appendix \ref{app_irrat}. 
\end{proof}

Given a point $(\omega_1,\omega_2)\in \cR$ we can define the irrationality measure $\mu(r)\in [1,\infty]$ of a ray $r\subset \bC^*$ as follows.
If $\pm \omega_2\in r$  we define $\mu(r)=1$. Otherwise there is a unique $\alpha\in \bR$ such that $\pm(\omega_1+\alpha\cdot \omega_2)\in r$  and we define $\mu(r)=\mu(\alpha)$. 
Part (iii) of Theorem \ref{roth}  ensures that the resulting notion   depends only on the lattice $\Lambda(\omega_1,\omega_2)\subset \bC$ rather than the specific generators $\omega_1,\omega_2$. Part (i) shows that $\mu(r)=1$ precisely if $r$ contains a lattice point, and part (iv) that almost all rays have $\mu(r)=2$.

\subsection{Borel sums and integrals along detour paths}\label{Borelssummary}

Let us again fix a point $(\omega_1,\omega_2)\in \cR$. Recall that a ray $r\subset \bC^*$  is non-Stokes precisely if it contains no lattice points.  The following results about Borel summation are proved in Section \ref{borelsummabilityproofsec}.

\begin{thm}\label{Borelsumthm}
    Let $r\subset \bC^*$ be a non-Stokes ray with  $\mu(r)<\infty$ and take $\epsilon \in \bH_r$.
  \begin{itemize}
     \item[(i)] 
The integrals \begin{equation}
    \label{transform2}
F^r(\epsilon \b \omega_1,\omega_2)=\int_r e^{-\eta/\epsilon} f(\eta\b \omega_1,\omega_2)\mathrm{d}\eta, \qquad K_i^r(\epsilon)=\int_r\ e^{-\eta/\epsilon} k_i(\eta \b \omega_1,\omega_2)\mathrm{d}\eta,\end{equation}
   are absolutely convergent and depend holomorphically on $\epsilon\in \mathbb{H}_r$. In particular, the series  $F(\epsilon\b \omega_1,\omega_2)$ and $K_i(\epsilon\b \omega_1,\omega_2)$ are Borel summable along the ray $r$.
   \item[(ii)] The  Borel sums can be re-expressed as absolutely convergent integrals
\begin{equation} \label{works}F^r(\epsilon\b\omega_1,\omega_2)=\int_{r} \mathrm{Li}_2(e^{- \eta/\epsilon}) \, \mathscr{f}(\eta \b \omega_1,\omega_2)\mathrm{d}\eta,
\end{equation}
\begin{equation}\label{depot}
K_{i}^r(\epsilon\b \omega_1,\omega_2)=\int_{r} \Li_1(e^{- \eta/\epsilon})\, \mathscr{k}_i (\eta\b \omega_1,\omega_2)\mathrm{d}\eta,\end{equation}
where $\mathrm{Li}_k(z)$ denotes the $k$-th polylogarithm. \qed
\end{itemize}
\end{thm}

Note that equation \eqref{works}  follows from \eqref{curlyf} and the following formal rearrangements, which are justified in the proof of Theorem \ref{Borelsumthm}: 

\begin{equation}
\begin{split}
\int_r e^{-\eta/\epsilon} \, \sum_{m\geq 1} \frac{1}{m^3}  \,  \mathscr{f}\Big(\frac{\eta}{ m}\Big) \, \mathrm{d}\eta
&=\sum_{m\geq 1}\, \int_r e^{-\eta/\epsilon} \, \frac{1}{m^3}\, \mathscr{f}\Big(\frac{\eta}{ m}\Big)
\, \mathrm{d}\eta=\sum_{m\geq 1}\, \int_{r} \frac{1}{m^2} \, e^{-m\eta/\epsilon}\mathscr{f}(\eta)\, \mathrm{d}\eta\\
&=\int_{r} \, \sum_{m\geq 1} \frac{1}{m^2} \, e^{-m\eta/\epsilon} \mathscr{f}(\eta)\, \mathrm{d}\eta=\int_{r} \Li_2(e^{-\eta/\epsilon}) \mathscr{f}(\eta)\, \mathrm{d}\eta.
\end{split}
\end{equation}
Similar remarks apply to \eqref{depot}.

 Consider now an arbitrary non-Stokes ray $r\subset \bC^*$.  For any $0<\delta\ll \text{min}\{|\omega_1|,|\omega_2|\}$  there is a uniquely-defined detour path $r(\delta)$ which combines the ray $r$ with arcs of angle $<\pi$ taken from  discs of radius $\delta$ centered on  points of $\Lambda^*(\omega_1,\omega_2)$ (see Figure \ref{detourpathfig}). The following proposition is proved in Section \ref{detourpathproof}.

 \begin{figure}
     \centering
     \includegraphics[scale=0.7]{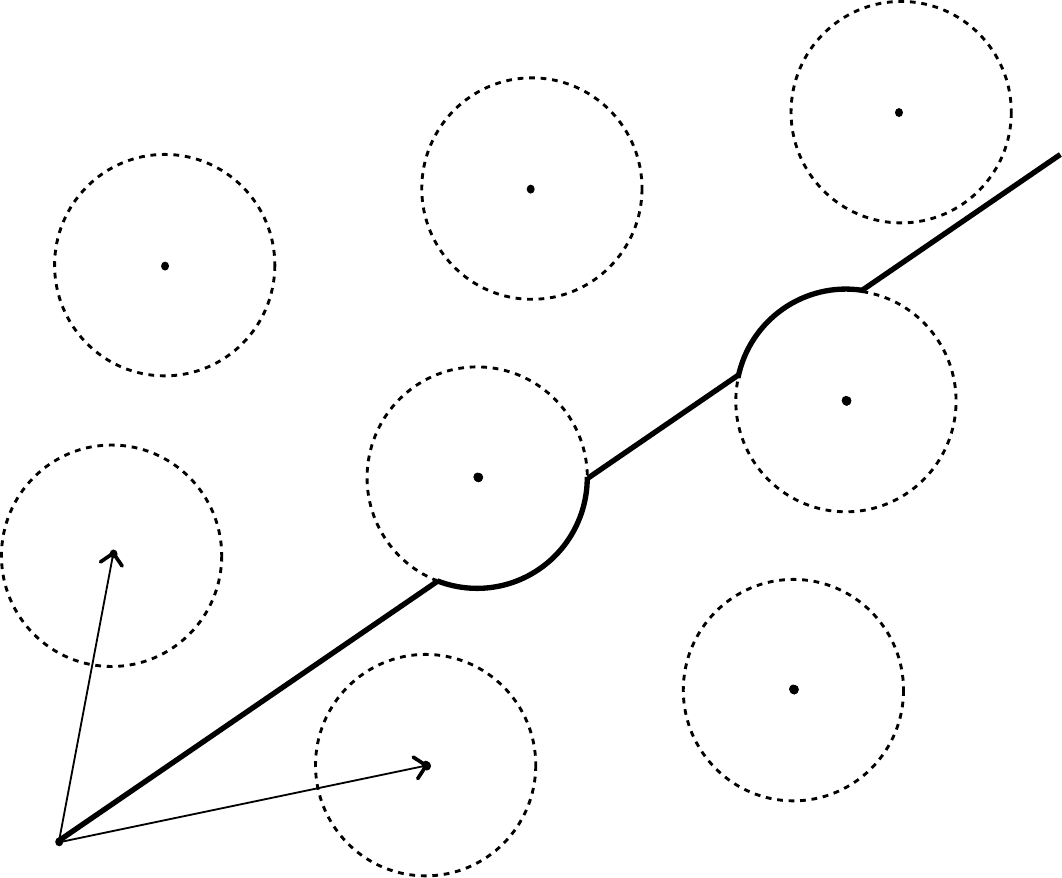}
     \caption{The vectors denote the generators $\omega_1$ and $\omega_2$ of the lattice $\Lambda(\omega_1,\omega_2)$. The discs are centered at the points in $\Lambda(\omega_1,\omega_2)^*$ and have radius $\delta>0$  small enough that they do not intersect. The bold path $r(\delta)$ is determined by the direction of the non-Stokes ray $r$, and  takes a detour along the boundary of any disc intersected by $r$. These detours traverse arcs of the boundary of angle $<\pi$.}
     \label{detourpathfig}
 \end{figure}

\begin{prop}
\label{last}
Let $r\subset \bC^*$ be a  non-Stokes ray with respect to the lattice $\Lambda(\omega_1,\omega_2)$. Then there is $D>0$  such that for all $0<\delta <D$ the integrals \begin{equation} \label{works2}F^{r(\delta)}(\epsilon\b\omega_1,\omega_2)=\int_{r(\delta)} \mathrm{Li}_2(e^{- \eta/\epsilon}) \, \mathscr{f}(\eta \b \omega_1,\omega_2)\mathrm{d}\eta,
\end{equation}
\begin{equation}\label{depot2}
K_{i}^{r(\delta)}(\epsilon\b \omega_1,\omega_2)=\int_{r(\delta)} \Li_1(e^{- \eta/\epsilon})\, \mathscr{k}_i (\eta\b \omega_1,\omega_2)\mathrm{d}\eta,\end{equation} are absolutely convergent for all $\epsilon \in \bH_r$. The resulting integrals depend holomorphically on $\epsilon \in \mathbb{H}_{r}$, and are independent of $\delta$. Moreover 
\begin{equation}
    F^{r(\delta)}(\epsilon\b \omega_1,\omega_2)=F^{r}(\epsilon\b \omega_1,\omega_2), \qquad K_{i}^{r(\delta)}(\epsilon\b \omega_1,\omega_2)=K_{i}^{r}(\epsilon\b \omega_1,\omega_2)
\end{equation}
whenever $\mu(r)<\infty$. \qed 
\end{prop}

Thus in the case of a non-Stokes ray $r\subset \bC^*$ satisfying $\mu(r)=\infty$ we can use \eqref{works2} and \eqref{depot2} to define substitutes for the functions  \eqref{transform2}, although these are no longer directly related  to the Borel sums of the series  $F(\epsilon\b \omega_1,\omega_2)$ and $K_i(\epsilon\b \omega_1,\omega_2)$.

Finally, we record how the integrals corresponding to different rays are related to each other. This proposition is proved at the end of Section \ref{stokesjumpssec}.

\begin{prop}\label{stokesjumpsprop}
    Let $\Delta\subset \bC^*$ be a convex sector whose boundary consists of two  rays $r_1$ and $r_2$ taken in clockwise order.  Assume the rays $r_1, r_2$ are non-Stokes  with respect to the lattice $\Lambda(\omega_1,\omega_2)$. Then for $\epsilon \in \mathbb{H}_{r_1}\cap \mathbb{H}_{r_2}$ and small enough $\delta>0$ we have
    
    \begin{equation}
    \begin{split}
        K_{i}^{r_2(\delta)}(\epsilon \b \omega_1,\omega_2)-K_{i}^{r_1(\delta)}(\epsilon \b \omega_1,\omega_2)&=2\pi \mathrm{i}\sum_{\omega\in \Delta\cap \Lambda^*(\omega_1,\omega_2)}a_i\cdot \log(1-e^{-\omega/\epsilon})\\
         F^{r_2(\delta)}(\epsilon,\omega_1,\omega_2)-F^{r_1(\delta)}(\epsilon,\omega_1,\omega_2)&=2\pi \mathrm{i}\sum_{\omega \in \Delta\cap \Lambda^*(\omega_1,\omega_2)}\frac{\partial}{\partial \epsilon}\left(\epsilon\,  \mathrm{Li}_2(e^{-\omega/\epsilon})\right)\,,
    \end{split}
    \end{equation}
    where we write $\omega=a_1\omega_1+a_2\omega_2$. \qed
\end{prop}


\section{Free energy, DT invariants and the RH problem}\label{appsec}

In this section we  use the analytic results of the previous section to prove our main results. As in the introduction we consider a smooth projective CY threefold with an elliptic fibration $\pi\colon X \to B$  satisfying Assumptions \ref{CYass}.

\subsection{Free energy and its Borel sums}
 Recall from \eqref{free_energy} that the $g\geq 2$ part of the GW generating function in the fibre classes is given by the formal power series in $\lambda$
 \begin{equation}
F_{\text{GW}}(\lambda\b \tau)=-e(X)\cdot \sum_{g\geq 2} \frac{B_{2g} \, G_{2g-2}(\tau)}{4g(2g-2)}\cdot \left(\frac{\lambda}{2\pi}\right)^{2g-2},\end{equation}
whose coefficients depend on a K{\"a}hler parameter $\tau\in \bC$ satisfying $\Im(\tau)>0$. As before, we set $2\pi \mathrm{i}\epsilon=\lambda/2\pi$, and via the change of variables \eqref{free_energyres}   consider $F_{\text{GW}}(\epsilon \b \omega_1,\omega_2)$  as a function of $\epsilon\in \bC^*$ and $(\omega_1,\omega_2)$ lying in the region
\begin{equation}\label{reg2}
    \cR=\{(\omega_1,\omega_2)\in (\bC^*)^2: \Im (\omega_2/\omega_1)>0\}.
\end{equation} Recall from \eqref{FGWFrelres} that $F_{\text{GW}}(\epsilon \b \omega_1,\omega_2)$ and $F(\epsilon \b \omega_1,\omega_2)$ are related by a rescaling by $-e(X)/(2\pi \mathrm{i})^2$. The following result then follows  immediately by combining Proposition \ref{propone} and Theorem \ref{Borelsumthm}. 

\begin{thm}
 \label{FGWBsum}
Fix $(\omega_1,\omega_2)\in \cR$.
\begin{itemize}
    \item[(i)] 
The Borel transform of the series $F_{\text{GW}}(\epsilon\b\omega_1,\omega_2)$  is a meromorphic function on $\bC$ with double poles at the non-zero lattice points $\Lambda^*(\omega_1,\omega_2)$ and no other poles. 

\item[(ii)]Suppose a non-Stokes ray $r\subset \bC^*$ satisfies $\mu(r)<\infty$ with respect to the lattice $\Lambda(\omega_1,\omega_2)$. Then the Borel sum $F^r_{\text{GW}}(\epsilon\b\omega_1,\omega_2)$ exists for all $\epsilon \in \mathbb{H}_r$. \qed
\end{itemize}
\end{thm}

Theorem \ref{FGWBsum} together with part (iv) of Theorem \ref{roth} implies that for almost all non-Stokes rays $r\subset \bC^*$ the Borel sum of $F_{\text{GW}}(\epsilon\b\omega_1,\omega_2)$ exists for $\epsilon\in \bH_r$.
Combining Proposition \ref{stokesjumpsprop} with \eqref{FGWFrelres} gives the following result relating the Borel sums along different rays.

\begin{prop}\label{stokesjumpsfreeenergy}
    Fix $(\omega_1,\omega_2)\in \cR$. Let $\Delta\subset \bC^*$ be a convex sector whose boundary consists of two  rays $r_1$ and $r_2$ taken in clockwise order.  Assume the rays $r_1, r_2$ are non-Stokes  and satisfy $\mu(r_i)<\infty$ with respect to the lattice $\Lambda(\omega_1,\omega_2)$.  Then
    \begin{equation}
         F^{r_2}_{\text{GW}}(\epsilon\b \omega_1,\omega_2)-F^{r_1}_{\text{GW}}(\epsilon\b\omega_1,\omega_2)=-\frac{e(X)}{2\pi \mathrm{i}}\cdot \sum_{\omega \in \Delta\cap \Lambda^*(\omega_1,\omega_2)}\frac{\partial}{\partial \epsilon}\left(\epsilon\,  \mathrm{Li}_2(e^{-\omega/\epsilon})\right)\,
    \end{equation}
    for all $\epsilon \in \mathbb{H}_{r_1}\cap \mathbb{H}_{r_2}$. \qed
\end{prop}
This matches previous results on the Stokes jumps of the Borel sum of free energies and their relation to DT invariants. See for example \cite[Equation 4.55]{ASTT} or \cite[Equation 1.1 and 1.4]{IM}.

\subsection{Stability conditions and DT invariants}\label{stabDTsec}

  We consider the full triangulated subcategory $\cD(\pi)\subset \D^b\Coh(X)$ of the bounded derived category of coherent sheaves  consisting of objects  whose set-theoretic support is contained in a finite union of fibres of $\pi$. The Chern characters of such objects can be viewed as elements 
\begin{equation}
    \ch(E)=(\ch_2(E),\ch_3(E))\in N(\pi)=N_1(\pi)\oplus N_0(X),
\end{equation}
where $N_1(\pi)\subset N_1(X)$ consists of curve classes contracted by $\pi$. The group $N_0(X)$ is freely generated by the class of a point, which it is convenient to denote by $-\gamma_1$. The assumption that $\pi$ has integral fibres implies that $N_1(\pi)$ is freely generated by the class $\gamma_2$  of a fibre. Then 
\begin{equation}\label{che}\ch\colon K_0(\cD(\pi))\to N(\pi)=\bZ\gamma_1\oplus \bZ\gamma_2,
\end{equation}
sends a  rank $r$, degree $d$ bundle supported on a smooth fibre of $\pi$ to the class $-d\gamma_1+r\gamma_2$. The Riemann-Roch theorem shows that for any objects $A,B\in \cD(\pi)$ we have
\[\chi(A,B):=\sum_{i\in \bZ} \dim_{\bC} \Hom_X(A,B[i])=0.\]
Thus the Euler form for the category $\cD(\pi)$ is identically zero, and we therefore also equip the group $N(\pi)$ with the zero form $\<-,-\>=0$.

The definition of the subcategory $\cD(\pi)$ ensures that the standard t-structure on $\cD^b\Coh(X)$ induces a t-structure on $\cD(\pi)$. The heart $\cA(\pi)\subset \cD(\pi)$ consists of coherent sheaves on $X$ whose set-theoretic support is contained in a finite union of fibres of $\pi$. 
Fix an element $\tau\in \bC$ with $\Im(\tau)>0$. Then, as in \cite[Example 2.3 (iii)]{T}, there is  a unique stability condition on the category $\cD(\pi)$ whose heart is the subcategory $\cA(\pi)\subset \cD(\pi)$, and whose central charge $Z\colon K_0(\cD(\pi))\to \bC$ is the composition of the Chern character \eqref{che} with the map
\[Z\colon N(\pi)\to \bC, \qquad Z(a_1\gamma_1+a_2\gamma_2)=a_1+a_2\tau.\]

There is a standard action of the group $\bC$ on the space of stability conditions which rotates the central charge and shifts  the phases of the semistable objects. Applying this to the stability conditions constructed above we obtain for each point $(\omega_1,\omega_2)\in \cR$,  a stability condition, unique up to the action of even shifts, whose central charge is the composition of \eqref{che} with the map
 \begin{equation}\label{z}
 Z\colon N(\pi)\to \bC, \qquad Z(a_1\gamma_1+a_2\gamma_2)=a_1\omega_1+a_2\omega_2.\end{equation}

Since rotating stability conditions does not effect the subcategories of semistable objects, the 
calculation of Toda \cite[Theorem 6.9]{T} shows that the DT invariants for any of these stability conditions are given by 
  \begin{equation}\Omega(a_1\gamma_1+a_2\gamma_2)=-e(X), \qquad (a_1,a_2)\in \bZ^2\setminus\{0\},\end{equation}
  where $e(X)$ is the topological Euler characteristic of the complex projective variety $X$.

\subsection{Riemann-Hilbert problem}

Fix a  point $(\omega_1,\omega_2)\in \cR$. The data  introduced in the previous section defines what is called a BPS structure in \cite{RHDT}. Namely we have a finite-rank free abelian group $N(\pi)$ equipped with a skew-symmetric form $\<-,-\>$, a group homomorphism $Z\colon N(\pi)\to \bC$, and a map of sets $\Omega\colon N(\pi)\to \bZ$ which encodes the DT invariants.  Following \cite{RHDT}, and exactly as in \cite{B}, we now explain the steps to go from this data to a RH problem.

Since the skew-symmetric form $\<-,-\>$ on $N(\pi)$ is identically zero, in order to obtain a non-trivial RH problem we must first perform the doubling procedure of \cite[Section 2.8]{RHDT}. To do this we introduce the dual abelian group $N(\pi)^{\vee}=\Hom_{\bZ}(N(\pi),\bZ)$ and consider the lattice $\Gamma=N(\pi)\oplus N(\pi)^{\vee}$ equipped with the canonical skew-symmetric pairing \begin{equation}\label{non-degp}\langle-,-\rangle\colon \Gamma\times \Gamma\to \bZ, \qquad \langle (\gamma,\lambda),(\gamma',\lambda')\rangle=\lambda(\gamma')-\lambda'(\gamma).\end{equation}
We denote by $\gammac_i$ the basis element of $N(\pi)^{\vee}$ dual to $\gamma_i$. Thus $\langle \gammac_i,\gamma_j\rangle=\delta_{ij}$.
We extend the central charge map $Z\colon N(\pi)\to \bC$ defined by \eqref{z} arbitrarily to a homomorphism $Z\colon \Gamma\to \bC$. The choice of this extension will play no significant role below. We also extend the map of sets $\Omega\colon N(\pi)\to \bZ$ to $\Gamma$ by insisting that $\Omega(\gamma)=0$ unless $\gamma\in N(\pi)\subset \Gamma$. 

The resulting doubled BPS structure $(\Gamma,Z,\Omega)$ has several special properties identified in \cite{RHDT}. It is convergent because for large enough $R> 0$
 \begin{equation}\sum_{(a_1,a_2)\in \bZ^2\setminus\{0\}} \exp({-R|a_1\omega_1+a_2\omega_2|})\, <\infty.\end{equation}
 It is moreover uncoupled since $\{\gamma \in \Gamma : \Omega(\gamma)\neq 0\}\subset N(\pi)$ and $\langle \gamma_1,\gamma_2\rangle=0$ for $\gamma_1,\gamma_2\in N(\pi)$.
 We can then formulate a RH problem  exactly as in \cite{B}. As well as the BPS structure $(\Gamma,Z,\Omega)$ it depends on an element $\xi$ of the twisted torus
 \begin{equation}\{\xi\colon \Gamma\to \bC^*:\xi(\gamma_1+\gamma_2)=(-1)^{\langle \gamma_1,\gamma_2\rangle}\xi(\gamma_1) \xi(\gamma_2)\}\end{equation}
 called the constant term.
 
 Recall from Section \ref{DTRHintrosec} that a ray $r\subset \bC^*$ is called a Stokes ray if it contains a point of $\Lambda(\omega_1,\omega_2)$, and otherwise a non-Stokes ray. 
 The  RH problem involves holomorphic functions $X_{\gamma}^{r}\colon \bH_r\to \bC^*$ for each non-Stokes ray $r$ and each class $\gamma\in \Gamma$, where as before  $\bH_r\subset \bC^*$ denotes the half-plane centered on $r$. 
 Arguing as in \cite[Section 5.1]{RHDT} we can use the fact that $(\Gamma,Z,\Omega)$ is an uncoupled BPS structure to write for each $i=1,2$ 
 \begin{equation}
 X_{\gamma_i}^{r}(\epsilon)=\exp(-Z(\gamma_i)/\epsilon)\cdot \xi(\gamma_i), \qquad X_{\gammac_i}^{r}=\exp(-Z(\gammac_i)/\epsilon)\cdot \xi(\gammac_i)\cdot Y_{i}^{r}(\epsilon),
 \end{equation}
 with $Y_{i}^{r}\colon \bH_r\to \bC^*$ holomorphic. For $i=1,2$ we define
 \begin{equation}\ell_i\colon \Lambda(\omega_1,\omega_2)\to \bZ, \qquad \ell_i(a_1\omega_1+a_2\omega_2)=a_i.\end{equation}
 
  Choose a constant term $\xi\colon \Gamma\to \bC^*$ satisfying $\xi(\gamma_i)=1$ for $i=1,2$.  Then the  RH problem can be formulated as follows:
 
 \begin{problem}
 For each non-Stokes ray $r\subset \bC^*$ find holomorphic functions $Y_{i}^{r}\colon \bH_r\to \bC^*$ with $i=1,2$ such that the following statements hold.
 \begin{itemize}
 \item[(RH1)] If $\Delta\subset \mathbb{C}^*$ is a convex sector whose boundary consists of non-Stokes rays $r_1,r_2$ taken in clockwise order then
 \begin{equation}
 \label{jumps2}Y_{i}^{r_2}(\epsilon)=Y_{i}^{r_1}(\epsilon)\cdot \prod_{\omega\in \Delta(r_1,r_2)\cap \Lambda^*}\left(1-e^{- \omega/\epsilon}\right)^{-\ell_i(\omega)\cdot e(X)},\end{equation}
for $\epsilon\in \bH_{r_1}\cap \bH_{r_2}$ with $0<|\epsilon|\ll 1$.

\item[(RH2)] As $\epsilon\to 0$ in any closed subsector of $\bH_r$ we have $Y^r(\epsilon)\to 1$.
 \item[(RH3)] There is an $N>0$ such that as $\epsilon\to \infty$ in $\bH_r$ there is a bound $|\epsilon|^{-N}< |Y^r(\epsilon)|<|\epsilon|^N$.
\end{itemize}
\end{problem}

If this problem has a solution then it is unique \cite{B}. We shall instead consider what we call the weak RH problem 
in which we drop condition (RH3). The resulting solutions are unique up to multiplication of $Y_{i}^{r}$ by  arbitrary holomorphic functions $P_i\colon \bC\to \bC^*$ satisfying $P_i(0)=1$. 
 
 \subsection{Solution to the weak RH problem}
We again fix a point $(\omega_1,\omega_2)\in \cR$. Recall the functions $K_i^{r(\delta)}(\epsilon \b \omega_1,\omega_2)$ defined in Proposition \ref{last}. For each non-Stokes ray $r\subset \bC^*$ we define a function $Y_{i}^{r}\colon \bH_r\to \bC^*$ by

\begin{equation}
\label{hear}
\begin{split}
    Y_{i}^r(\epsilon)&:=\exp\Big( -\frac{e(X)}{2\pi \mathrm{i}} \cdot K_{i}^{r(\delta)}(\epsilon\b\omega_1,\omega_2)\Big)\\
    &=\exp\left(-\frac{e(X)}{2\pi\mathrm{i}}\int_{r(\delta)} \mathrm{Li}_1(e^{- \eta/\epsilon})\,\mathscr{k}_i(\eta\b\omega_1,\omega_2)\mathrm{d}\eta\right)\,.
\end{split}
\end{equation}
The integral is absolutely convergent, holomorphic in $\epsilon \in \mathbb{H}_{r}$, and does not depend on $0< \delta\ll 1$ by Proposition \ref{last}. As before, we remark that for almost all non-Stokes rays $r$ we have $\mu(r)<\infty$, and for such rays it follows from Proposition \ref{last} and Theorem \ref{Borelsumthm} that 
\begin{equation}\label{wRHsol}
    Y_{i}^r(\epsilon)=\exp\left(-\frac{e(X)}{2\pi \mathrm{i}}\int_r e^{-\eta/\epsilon}\, k_i(\eta\b\omega_1,\omega_2)\mathrm{d}\eta\right)\,.
\end{equation}

 \begin{thm}\label{RHthm} The functions $Y_{i}^r\colon \bH_r\to \bC^*$ give a solution to the weak RH problem. 
\end{thm}

 \begin{proof}
Let  $\Delta\subset \mathbb{C}^*$ be a convex sector whose boundary consists of non-Stokes rays $r_1,r_2$ taken in clockwise order. By Proposition \ref{stokesjumpsprop} it follows that for $\epsilon \in \mathbb{H}_{r_1}\cap \mathbb{H}_{r_2}$
 \begin{equation}
 \begin{split}
     Y_{i}^{r_2}(\epsilon)&=Y_{i}^{r_1}(\epsilon)\cdot \exp\left(-e(X)\sum_{\omega\in \Delta\cap \Lambda^*}\ell_i(\omega)\cdot \log(1-e^{-\omega/\epsilon})\right)\\
     &=Y_{i}^{r_1}(\epsilon)\prod_{\omega\in \Delta\cap \Lambda^*}(1-e^{-\omega/\epsilon})^{-e(X)\cdot \ell_i(\omega)}
\end{split}
 \end{equation}
 so property (RH1) holds. Property (RH2) follows from the following lemma. 
 \end{proof}

 \begin{lemma}\label{rh2}
    Fix a non-Stokes ray $r\subset \bC^*$ and a closed subsector  $S_r\subset \mathbb{H}_r$. Then 
    \begin{equation}
        \lim_{\epsilon \to 0, \;\epsilon \in S_r}K^{r(\delta)}_i(\epsilon\b \omega_1,\omega_2)=0.
    \end{equation}
\end{lemma}
\begin{proof}
    Given the closed sector $S_r$, we can assume that $\delta>0$ is small enough such that 
    \begin{equation}
        \text{Re}(\eta/\epsilon)>0 \text{ for all $\eta\in r(\delta)$ and all  $\epsilon \in S_r$\,.}
    \end{equation}
   More precisely, given $K>0$  there exists a constant $C>0$ such that 
    \begin{equation}
        \text{Re}(\eta/\epsilon)=\left|\frac{\eta}{\epsilon}\right|\cos(\arg(\eta/\epsilon))>C\cdot \left|\frac{\eta}{\epsilon}\right|>\frac{C}{K}\cdot |\eta|>0
    \end{equation}
    for all $\eta \in r(\delta)$ and $\epsilon \in S_r$ with $|\epsilon|<K$. This in particular implies (using that $|\log(1-z)|<-\log(1-|z|)$ for $|z|<1|$) that in the same range of parameters
    \begin{equation}
        |\mathrm{Li_1}(e^{-\eta/\epsilon})\, \mathscr{k}_i(\eta)|<-\log(1-|e^{-\eta/\epsilon}|)\, |\mathscr{k}_i(\eta)|<-\log(1-e^{-C|\eta|/K})\, |\mathscr{k}_i(\eta)|\,.
    \end{equation}
   
    Finally, by the same argument as in Proposition \ref{weakRHsol} one can show that 
    \begin{equation}
        -\int_{r(\delta)}\log(1-e^{-C|\eta|/K})|\mathscr{k}_i(\eta)||\mathrm{d}\eta|<\infty,
    \end{equation}
    so by applying the dominated convergence theorem, we can interchange limits and integrals and obtain
    \begin{equation}
    \begin{split}
         \lim_{\epsilon \to 0, \;\epsilon \in S_r}K_{i}^{r(\delta)}&=\int_{r(\delta)}\mathrm{Li}_1(0)\mathscr{k}_i(\eta,\omega_1,\omega_2)\mathrm{d}\eta=0\,.\\
    \end{split}
    \end{equation}
\end{proof}

Fix a non-Stokes ray $r\subset \bC^*$. For  $\epsilon \in \mathbb{H}_r$ we define
 \begin{equation}\label{taurdef}
    \tau^r_{\text{GW}}(\epsilon \b \omega_1,\omega_2):=\exp\Big(-\frac{e(X)}{(2\pi\mathrm{i})^2}\cdot F^{r(\delta)}(\epsilon\b \omega_1, \omega_2)\Big)\,,
\end{equation}
where  $0<\delta\ll 1$  and 
$F^{r(\delta)}(\epsilon\b \omega_1, \omega_2)$ is as in Proposition \ref{last}. Note  that by \eqref{FGWFrelres}
\begin{equation}
    \tau^r_{\text{GW}}(\epsilon \b \omega_1, \omega_2)=\exp\Big(F^{r(\delta)}_{\text{GW}}(\epsilon \b \omega_1,\omega_2)\Big)\,.
\end{equation}
In the case that $\mu(r)<\infty$ the following result relates the Borel sum of the free energy to the solution to the weak RH problem constructed above.

\begin{thm}
    For each $i=1,2$ there is a relation
    \begin{equation}\label{tauRH}
        \frac{\partial}{\partial \omega_i}\log\tau^r_{\text{GW}}(\epsilon \b \omega_1,\omega_2)=\frac{1}{2\pi \mathrm{i}}\cdot \frac{\partial}{\partial \epsilon}\log Y_{i}^{r}(\epsilon \b \omega_1,\omega_2)\,.
    \end{equation}
    \begin{proof}
        Using \eqref{taurdef} and \eqref{wRHsol}, showing that \eqref{tauRH} holds reduces to showing that 
        \begin{equation}
            \frac{\partial}{\partial \omega_i}F^{r(\delta)}(\epsilon \b \omega_1,\omega_2)=\frac{\partial}{\partial \epsilon} K_i^{r(\delta)}(\epsilon \b \omega_1,\omega_2)\,.
        \end{equation}
        This easily follows by differentiating under the integral sign and integrating by parts twice using \eqref{mod_zeta}, \eqref{mod_wp}, \eqref{curlydef1}, and \eqref{curlydef2}.
    \end{proof}
\end{thm}

\section{Borel transforms}\label{Boreltranssec}

In this section we prove Proposition \ref{propone} concerning the Borel transforms of our series. We fix a point  $(\omega_1,\omega_2)\in \cR$ throughout. 

\subsection{Hadamard product}
Consider again the power series
\begin{equation}H(\epsilon)=\sum_{g\geq 2} \frac{B_{2g} \,G_{2g-2}\, (2\pi \mathrm{i})^{2g}\epsilon^{2g-1}}{4g(2g-1)(2g-2)} \in \epsilon\bC[[\epsilon]].\end{equation}
 from \eqref{power1}. We have suppressed the dependence  on $\omega_1,\omega_2$ from the notation. The Borel transform is the series
\begin{equation}h(\eta)=\sum_{g\geq 2} \frac{B_{2g} \, G_{2g-2}\, (2\pi \mathrm{i})^{2g}\eta^{2g-2}}{2(2g-2)(2g)!}\in \bC[[\eta]].\end{equation}
 Following the approach of \cite[Section 3.1]{ASTT} we can write $h(\eta)$ as a Hadamard product of  series $h_1(\eta)$ and $h_2(\eta)$, where
\begin{equation}
\label{h1h2}h_1(\eta)=-\sum_{g\geq 2} \frac{B_{2g} \, (2\pi \mathrm{i})^{2g}\eta^{2g-2}}{2 (2g)!}, \qquad h_2(\eta)=-\sum_{g\geq 2} \frac{G_{2g-2}\,  \eta^{2g-2}}{2g-2}.\end{equation}

Using the defining generating series for the Bernoulli numbers we find that $h_1(\eta)$ is the Taylor expansion at the origin of a meromorphic function on $\bC$ with poles only at the points  $m\in \bZ\setminus\{0\}$. Indeed
\begin{equation}\label{h1sum}
    h_1(\eta)=(2\pi \mathrm{i})^2\left(\frac{1}{2(2\pi \mathrm{i}\eta)(1-e^{2\pi \mathrm{i}\eta})}+\frac{1}{2(2\pi \mathrm{i}\eta)^2} -\frac{1}{4(2\pi \mathrm{i}\eta)} +\frac{1}{24}\right).
\end{equation}

Moreover $h_2(\eta)$ is  the Taylor expansion at $\eta=0$ of the  function 
\begin{equation}
    \label{mod_h2}
\mathscr{h}(\eta)=\log \sigma(\eta)-\log(\eta)-\half G_2 \eta^2\end{equation}
introduced in Section \ref{elliptic}.
It follows from \cite[Lemma 3.2]{ASTT} that $h(\eta)$ is the Taylor expansion at $\eta=0$ of the function
$h(\eta)$ given by the anti-clockwise contour integral
\begin{equation}
    \label{countdown}
h(\eta)=\frac{1}{2\pi \mathrm{i}} \int_{ |s|=\tfrac{1}{2}} h_1(s)\,  h_2(\eta/s) \, \frac{\mathrm{d}s}{s}.\end{equation}
This expression is valid providing $|\eta|<\tfrac{1}{2} |\omega|$ for all nonzero lattice points $\omega \in \Lambda^*(\omega_1,\omega_2)$. This ensures that  the singularities of $h_2(\eta/s)$ all lie inside the contour $|s|=\tfrac{1}{2}$. Note that  the poles  of $h_1(s)$ always lie outside this contour.

\subsection{Borel transforms}
  Recall the  formal power series $F(\epsilon)$ and $K_i(\epsilon)$ defined by \eqref{power2}. Reall also the functions \begin{equation}\label{curlydef12}\mathscr{f}(\eta\b\omega_1,\omega_2)=2\tilde{\zeta}(\eta\b\omega_1,\omega_2)-\eta \tilde{\wp}(\eta\b\omega_1,\omega_2), \end{equation}
\begin{equation}\label{curlydef22}\mathscr{k}_i(\eta\b\omega_1,\omega_2)=\frac{\partial}{\partial \omega_i} \mathscr{h} (\eta\b\omega_1,\omega_2).\end{equation}
defined in Section \ref{elliptic}. They are holomorphic in a neighbourhood of $\eta=0$. 

\begin{prop}
\label{count}
    The Borel transforms of the  series $F(\epsilon)$ and $K_i(\epsilon)$ are the Taylor expansions  of  holomorphic functions $f(\eta)$ and $k_i(\eta)$ defined near $\eta=0$  by the expressions
   \begin{equation}
    \label{inte}f(\eta)=\frac{1}{2\pi \mathrm{i}} \int_{ |s|=\tfrac{1}{2}} h_1(s)\,  \mathscr{f}(\eta/s) \, \frac{\mathrm{d}s}{s^2}, \qquad k_i(\eta)=\frac{1}{2\pi \mathrm{i}} \int_{ |s|=\tfrac{1}{2}} h_1(s)\,  \mathscr{k}_i(\eta/s) \, \frac{\mathrm{d}s}{s}.\end{equation}
\end{prop}

\begin{proof}The expression for $k_i(\eta)$ follows immediately by differentiating \eqref{countdown} under the integral with respect to $\omega_i$. To obtain the expression for $f(\eta)$ 
   note first that if formal series $F(\epsilon), H(\epsilon)\in \epsilon\bC[[\epsilon]]$ have Borel transforms $f(\eta), h(\eta)\in \bC[[\eta]]$ respectively, then
\begin{equation}F(\epsilon)=\frac{\mathrm{d}}{\mathrm{d}\epsilon} H(\epsilon) \implies f(\eta)=\frac{1}{\eta} \frac{\mathrm{d}}{\mathrm{d}\eta} \bigg(\eta^2 \,\frac{\mathrm{d}}{\mathrm{d}\eta}h(\eta)\bigg).\end{equation}
Indeed, it is enough to check the case $H(\epsilon)=\epsilon^n$ when the relation becomes
\begin{equation} \frac{n\,\eta^{n-2}}{(n-2)!}=\frac{1}{\eta} \,\frac{\mathrm{d}}{\mathrm{d}\eta} \bigg(\eta^2 \frac{\mathrm{d}}{\mathrm{d}\eta}\bigg(\frac{\eta^{n-1}}{(n-1)!}\bigg)\bigg).\end{equation}
Next note that \eqref{mod_h} - \eqref{mod_wp} give
\begin{equation}\frac{1}{\eta} \frac{\mathrm{d}}{\mathrm{d}\eta} \bigg(\eta^2 \,\frac{\mathrm{d}}{\mathrm{d}\eta}\mathscr{h}\left(\frac{\eta}{s}\right)\bigg)=\frac{1}{s}\left(2\tilde{\zeta}\left(\frac{\eta}{s}\right)-\frac{\eta}{s}\tilde{\wp}\left(\frac{\eta}{s}\right)\right)=\frac{1}{s}\mathscr{f}\left(\frac{\eta}{s}\right).\end{equation}
The result then follows by differentiating \eqref{countdown} under the integral with respect to $\eta$. 
\end{proof}

\subsection{Explicit expressions}

The following result completes the proof of Proposition \ref{propone}.

\begin{prop}
    \label{down}
The functions $f(\eta)$ and $k_i(\eta)$  extend to meromorphic functions of $\eta\in \bC$ with poles precisely at the nonzero lattice points $\Lambda^*(\omega_1,\omega_2)$. These poles are double poles in the case of $f$ and simple poles in the case of $k_i$.
There are explicit expressions
\begin{equation}
    \label{curly}
f(\eta)=\sum_{m\geq 1} \frac{1}{m^3} \cdot \mathscr{f}\Big(\frac{\eta}{m}\Big), \qquad 
k_i(\eta)=\sum_{m\geq 1} \frac{1}{m^2} \cdot \mathscr{k}_i\Big(\frac{\eta}{ m}\Big), \end{equation}
where the two series converge uniformly and absolutely in $\eta$ on compact subsets of $\bC$.
\end{prop}

\begin{proof}
We show the result for $f$, since a similar argument applies to $k_i$.  For each integer $N>0$ we consider the square contour
\begin{equation} C_{N}=\big\{s\in \bC: \max \big(|\text{Re}(s)|,|\text{Im}(s)|\big)=N+\tfrac{1}{2}\big\},\end{equation}
taken with the anti-clockwise orientation.
Take $\eta\in \bC$ such that $|\eta|<\tfrac{1}{2}|\omega|$ for all nonzero lattice points $\omega\in \Lambda^*(\omega_1,\omega_2)$.

Note that the function $h_1(s)$ has a simple pole at each  point $m\in \mathbb{Z}\setminus\{0\}$ with residue $-1/2m$. Moving the contour in \eqref{inte} and using $\mathscr{f}(-\eta)=-\mathscr{f}(\eta)$  therefore shows that for any  integer $N>0$
\begin{equation}\label{parres}
f(\eta)-\frac{1}{2\pi \mathrm{i}} \int_{C_{N}} h_1(s)\,  \mathscr{f}(\eta/s) \, \frac{\mathrm{d}s}{s^2}=\sum_{m=1}^N \frac{1}{m^3} \cdot \mathscr{f}\Big(\frac{\eta}{m}\Big).
\end{equation}
On the other hand, one easily checks using the power expansion of $\mathscr{f}(\eta)$ at $\eta=0$ that 
\begin{equation}
    \mathscr{f}(\eta/s)=\mathcal{O}(1/s), \quad \text{as $|s|\to \infty$},
\end{equation} 
while 
\begin{equation}
    \left|\frac{1}{2(2\pi \mathrm{i} s)(1-e^{2\pi \mathrm{i}s})}\right|<\frac{C}{N+\tfrac{1}{2}}, \quad \text{for $s\in C_N$,}
\end{equation}
for some $C>0$ independent of $N$. It then follows from \eqref{h1sum} that there is some constant $D>0$ such that for all $N$ sufficiently large
\begin{equation}
    \left|\,\int_{C_{N}} h_1(s)\,  \mathscr{f}(\eta/s) \, \frac{\mathrm{d}s}{s^2}\,\right|<\frac{D}{(N+\tfrac{1}{2})^{2}}\,.
\end{equation}
Thus the integral on the left-hand side of \eqref{parres} tends to $0$ as $N\to \infty$, and \eqref{curly} holds. 

It remains to show that the series \eqref{curly} defines a meromorphic function on $\bC$ with double poles exactly at the points of  $\Lambda^*(\omega_1,\omega_2)$.  Note that $\mathscr{f}(\eta)$ has poles only at the points $\Lambda^*(\omega_1,\omega_2)$ and in particular is holomorphic at $\eta=0$.  

Let $D\subset \mathbb{C}$ be  small disc centered at $0$ such that $D\cap\Lambda^*(\omega_1,\omega_2)=\emptyset$ and $K\subset \bC$ a compact subset. For $\eta\in K$ and $M>0$ sufficiently large we have  
$\eta/m\in D$  for $m\geq M$.
In particular, the functions $\mathscr{f}(\eta/m)$   are holomorphic and uniformly  bounded for $m\geq M$ and $\eta \in K$. It follows that the tail of the first sum in \eqref{curly} (i.e. the sum for $m\geq M$) converges uniformly and absolutely on $K$, and hence to a holomorphic function on $K$. Since any lattice point $\omega \in \Lambda^*(\omega_1,\omega_2)$ is a double pole of the function $\mathscr{f}(\eta/m)$ for a finite but non-empty set of positive integers $m$, the resulting $f(\eta)$ is a meromorphic function with double poles at $\Lambda^*(\omega_1,\omega_2)$. 
\end{proof}


\section{Borel summability and detour integrals}\label{proofssec}

In this section we collect the results needed to show that the series $K_i(\epsilon\b \omega_1,\omega_2)$ and $F(\epsilon\b \omega_1,\omega_2)$ are Borel summable along almost all non-Stokes rays $r\subset \bC^*$. More precisely, the Borel sum exists for non-Stokes rays $r$ that, in the sense of Section \ref{Borelssummary}, have finite irrationality measure $\mu(r)<\infty$ with respect to the lattice $\Lambda(\omega_1,\omega_2)$. For a general non-Stokes ray $r$, we define  integrals along certain detour paths $r(\delta)$ which coincide with the Borel sums when $\mu(r)<\infty$. 

\subsection{Key lemmas}
We begin with the  following useful lemmas:

\begin{lemma}\label{lemma2} Fix $(\omega_1,\omega_2)\in \cR$ and take $\omega=a_1\omega_1+a_2\omega_2\in \Lambda(\omega_1,\omega_2)$. Then the functions  $\wp(\eta\b\omega_1,\omega_2)$, $\zeta(\eta\b \omega_1,\omega_2)$ and $\rho_i(\eta \b \omega_1,\omega_2):=\partial_{\omega_i}\log(\sigma(\eta \b \omega_1,\omega_2))$ satisfy

\begin{equation}\label{eq:periodicity}
    \begin{split}
\wp(\eta+\omega\b\omega_1,\omega_2)&=\wp(\eta\b\omega_1,\omega_2),\\
\zeta(\eta+\omega\b\omega_1,\omega_2)&=\zeta(\eta\b\omega_1,\omega_2)+2(a_1\eta_1+a_2\eta_2), \quad \text{where} \quad \eta_i=\zeta(\omega_i/2\b\omega_1,\omega_2),\\
        \rho_i(\eta+\omega\b \omega_1,\omega_2)&=\rho_i(\eta\b \omega_1,\omega_2)-a_i\zeta(\eta\b \omega_1,\omega_2)+\sum_{j=1,2}a_j(-a_i\eta_j+(2\eta+\omega)\partial_{\omega_i}\eta_j)\,.\\
    \end{split}
\end{equation}
\end{lemma}
\begin{proof}
    The first and second identity follow from the well-known periodicity of $\wp$ and quasi-periodicity of $\zeta$. On the other hand, the $\sigma$ function satisfies
    \begin{equation}\label{sigmaper}
    \sigma(\eta+\omega\b \omega_1,\omega_2)=(-1)^{a_1+a_2+a_1a_2}e^{(2\eta+\omega)(a_1\eta_1+a_2\eta_2)}\sigma(\eta\b\omega_1,\omega_2)\,,
\end{equation}
from which the last identity follows by taking logs and  derivatives with respect to $\omega_i$. Indeed, differentiating the left-hand side of \eqref{sigmaper} gives
\begin{equation}
\partial_{\omega_i}\log(\sigma(\eta+\omega \b \omega_1,\omega_2))=a_i\zeta(\eta+\omega\b \omega_1,\omega_2)+\rho_i(\eta+\omega\b \omega_1,\omega_2)\,,
\end{equation}
while differentiating the right hand side gives
\begin{equation}
    \partial_{\omega_i}\left((2\eta+\omega)(a_1\eta_1+a_2\eta_2)+\log(\sigma(\eta\b \omega_1,\omega_2))\right)=\rho_i(\eta\b \omega_1,\omega_2)+\sum_{j=1}^2a_j(a_i\eta_j +(2\eta+\omega)\partial_{\omega_i}\eta_j)\,.
\end{equation}
The last identity of \eqref{eq:periodicity} then follows by applying the second identity  and reorganizing terms. 
\end{proof}

Recall the notion of the irrationality measure $\mu(\alpha)$ of a real number $\alpha\in \mathbb{R}$, and of the irrationality measure $\mu(r)$ of a ray $r\subset \bC^*$ with respect to a lattice $\Lambda(\omega_1,\omega_2)$ introduced in Section \ref{Borelssummary}. In the following, we will use that if $\alpha\in \bR\setminus \bQ$ and $n>\mu(\alpha)$ then 
\begin{equation}\label{keyinq}
    |\alpha-p/q|\geq 1/q^n
\end{equation}
for all $p,q\in \mathbb{Z}$ with $q$ sufficiently large. 
\begin{lemma}\label{lemma3} Fix $(\omega_1,\omega_2)\in \cR$ and let $r\subset \bC^*$ be a non-Stokes ray such that $\mu(r)<\infty$ with respect to the lattice $\Lambda(\omega_1,\omega_2)$. Then for $\epsilon \in \mathbb{H}_{r}$ the  integrals 
\begin{equation}\label{eq:finiteint}
    \int_{r} e^{-\eta/\epsilon}\mathscr{f}(\eta\b \omega_1,\omega_2)\mathrm{d}\eta, \quad \int_{r} e^{-\eta/\epsilon}\mathscr{k}_i(\eta\b \omega_1,\omega_2)\mathrm{d}\eta,
\end{equation}
are absolutely convergent.
\end{lemma}
\begin{proof}
    The functions $\mathscr{f}(\eta)$, $\mathscr{k}_i(\eta)$ are meromorphic in $\eta$ with poles only at the nonzero lattice points  $\Lambda^*(\omega_1,\omega_2)$. 
    The fact that $r$ is a non-Stokes ray implies that $r$ can be parameterised as 
    \begin{equation}
        \eta(t)=\pm t(\omega_1+\alpha \omega_2) \,,
    \end{equation}
    with $t\in \bR_{\geq 0}$ and $\alpha \in \mathbb{R}\setminus \mathbb{Q}$. We assume that we are in the case with the $+$ sign,  with the other case being completely analogous. 
    
    Fix $K>0$. Since the  functions $\mathscr{f}$, $\mathscr{k}_i$ are holomorphic at $0\in \mathbb{C}$, the integrals over $t\in [0,K]$ are finite. 
    Now note that by \eqref{mod_h}-\eqref{mod_wp}, the functions $\tilde{\wp}$, $\tilde{\zeta}$ and $\mathscr{k}_i$ differ from $\wp$, $\zeta$ and $\rho_i$  by terms that have at most polynomial growth in $\eta(t)$ as $t\to \infty$. Due to the exponential decay of $e^{-\eta(t)/\epsilon}$ for $\epsilon \in \mathbb{H}_{r}$ as $t\to \infty$,  to show that \eqref{eq:finiteint} holds it is then  enough to check that 
    \begin{equation}
    \int_{r_{\infty}}|e^{-\eta/\epsilon}\wp(\eta)\mathrm{d}\eta|<\infty, \quad \int_{r_{\infty}} |e^{-\eta/\epsilon}\zeta(\eta)\mathrm{d}\eta|<\infty, \quad \int_{r_{\infty}} |e^{-\eta/\epsilon}\rho_i(\eta)\mathrm{d}\eta|<\infty,
\end{equation}
where $r_{\infty}$ is the segment given by $\eta(t)$ for $t\in [K,\infty)$. We start with the first of these statements.

We work with the inner product on $\bC$ in which  $\omega_1$ and $\omega_2$ are orthonormal and consider discs $D_\delta(\omega)$ of radius $0<\delta<K$   centered at the points of $\Lambda(\omega_1,\omega_2)$. We denote the norm induced by this inner product by $||
\cdot ||$ to distinguish it from the canonical norm $|
\cdot|$. We take  $\delta>0$ sufficiently small so that these discs do not intersect each other. Subdivide the ray $r_{\infty}$ into two sets 
    \begin{equation}
        r_{\infty}=r_p\cup r_c
    \end{equation}
    where $r_p$ is made up of the segments of $r_{\infty}$ inside the discs, and $r_{c}$ is the complement of $r_p$ in $r_{\infty}$. In particular, we can write
    \begin{equation}
        r_p=\bigcup_{\omega\in \Lambda^*(\omega_1,\omega_2)}r_{\omega}
    \end{equation}
    where $r_{\omega}$ is the segment contained in the disc $D_\delta(\omega)$. 
    
    Take a non-zero lattice point $\omega=a_1\omega_1+a_2\omega_2\in \Lambda^*(\omega_1,\omega_2)$ and consider
    \begin{equation}
        \int_{r_{\omega}} |e^{- \eta/\epsilon} \wp(\eta\b\omega_1,\omega_2)\mathrm{d}\eta|\,.
    \end{equation}
     Using the fact that $\wp$ is periodic for the lattice $\Lambda(\omega_1,\omega_2)$ and the Laurent expansion of $\wp(\eta)$ at $\eta=0$ we know  that  when $\eta\in D_\delta(\omega)$
    \begin{equation}
        \wp(\eta)=\wp(\eta-\omega)=\frac{1}{(\eta-\omega)^2}+\text{Reg}(\eta-\omega),
    \end{equation}
    where $\text{Reg}$ is a holomorphic function in the disc $D_\delta(0)$. So in particular we have 
    \begin{equation}
        |\wp(\eta)|=|\wp(\eta-\omega)|\leq \frac{1}{|\eta-\omega|^2}+D_1\leq \frac{1+C\cdot D_1\delta^2}{|\eta-\omega|^2}=\frac{D_2}{|\eta-\omega|^2}\,,
    \end{equation}
     where  $D_1>0$, $C>0$ is such that $|
     \cdot |<C||\cdot ||$, and  $D_2=1+C\cdot D_1\delta^2>0$ are constants independent of $\omega$. Since the canonical norm $|\cdot|$ is equivalent to $||\cdot||$, it follows that   
    \begin{equation}
        |\eta(t)-\omega|^2\geq D_3((t-a_1)^2 +(t\alpha -a_2)^2)
    \end{equation}
    for some $D_3>0$. Minimizing the right hand side we find

    \begin{equation}
        (t-a_1)^2 +(t\alpha -a_2)^2\geq \frac{(a_1\alpha -a_2)^2}{1+\alpha^2}\,.
    \end{equation}
    By picking $n>\mu(\alpha)$ and possibly increasing $K>0$, we can assume that for all $\omega\in \Lambda^*(\omega_1,\omega_2)$ such that $r_{\omega}$ is non-empty we have
    \begin{equation}
        |\alpha-a_2/a_1|\geq\frac{1}{|a_1|^n}\,,
    \end{equation}
    where we again wrote $\omega=a_1\omega_1+a_2\omega_2$. 
    Hence, overall on $r_{\omega}$ we have 
    \begin{equation}
        |\wp(\eta)|\leq \frac{D_2 (1+\alpha^2)\cdot |a_1|^{2n-2}}{D_3}\,.
    \end{equation}
    
    Recall that the discs $D_\delta(\omega)$ are defined with respect to the inner product where $\omega_1$ and $\omega_2$ are orthonormal. The points of intersection of the ray $\eta(t)$ with the boundary of  $D_\delta(\omega)$  occur when
    \begin{equation}
        t=t_{\pm}(\omega)=\frac{(a_1+\alpha a_2)\pm \sqrt{(1+\alpha^2)\delta^2-(\alpha a_1-a_2)^2}}{(1+\alpha^2)}\,.
    \end{equation}
    Note that if $r_{\omega}$ is not empty, we must have 
    \begin{equation}
        (1+\alpha^2)\delta^2-(\alpha a_1-a_2)^2=(1+\alpha^2)(\delta^2-\text{dist}(r,\omega)^2)\geq 0\,,
    \end{equation}
    where $\text{dist}(r,\omega)$ denotes the distance between $r$ and $\omega$ in the norm where $\omega_1$ and $\omega_2$ are orthonormal. Hence, we obtain
    \begin{equation}
        \int_{r_{\omega}}| e^{- \eta/\epsilon} \wp(\eta\b\omega_1,\omega_2)\mathrm{d}\eta|\leq \delta_1|a_1|^{2n-2}\int_{t_{-}}^{t_{+}}e^{-t\delta_2}\mathrm{d}t=\delta_1|a_1|^{2n-2}\left(-\frac{1}{\delta}_2(e^{-t_{+}\delta_2}-e^{-t_{-}\delta_2})\right)
    \end{equation}
    where 
    \begin{equation}
        \delta_1=\frac{D_2 (1+\alpha^2)|\omega_1+\alpha\omega_2|}{D_3}>0, \quad \delta_2 = \left|\frac{\omega_1+\alpha\omega_2}{\epsilon}\right|\cos\left(\text{Arg}\left(\frac{\omega_1+\alpha\omega_2}{\epsilon}\right)\right)>0\,,
    \end{equation}
    so $\delta_1,\delta_2$ are constants depending only on $\alpha,\omega_1,\omega_2,\delta,\epsilon$. Furthermore, note that 
    \begin{equation}
       e^{-t_{-}\delta_2}-e^{-t_{+}\delta_2}=2e^{-\delta_2\frac{a_1+\alpha a_2}{1+\alpha^2}}\sinh\left(\delta_2\frac{\sqrt{\delta^2-\text{dist}(r,\omega)^2}}{(1+\alpha^2)^{1/2}}\right)\leq 2e^{-\delta_2\frac{a_1+\alpha a_2}{1+\alpha^2}}\sinh\left(\frac{\delta_2\delta}{(1+\alpha^2)^{1/2}}\right)\,,
    \end{equation}
    so that 
    \begin{equation}
        \int_{r_{\omega}} |e^{- \eta/\epsilon} \wp(\eta\b\omega_1,\omega_2)\mathrm{d}\eta|\leq 2\frac{\delta_1}{\delta_2}|a_1|^{2n-2}\sinh\left(\frac{\delta_2\delta}{(1+\alpha^2)^{1/2}}\right)e^{-\delta_2\frac{a_1+\alpha a_2}{1+\alpha^2}}\,.
    \end{equation}
    Now note that if $\eta(t)=t(\omega_1+\alpha\omega_2)$ intersects the disc centered at $\omega=a_1\omega_1+a_2\omega_2$, and $\omega$ is sufficiently large in norm, then we must have $a_1,\alpha a_2>0$.  In particular, we find that 

    \begin{equation}
    \begin{split}
        \int_{r_{p}} |e^{- \eta/\epsilon} &\wp(\eta\b \omega_1,\omega_2)\mathrm{d}\eta|=\sum_{\omega\in \Lambda^*(\omega_1,\omega_2)}\int_{r_{\omega}} |e^{-\eta/\epsilon} \wp(\eta \b \omega_1,\omega_2)\mathrm{d}\eta|\\
        &<\sum_{(a_1,a_2)\in \mathbb{Z}^2\; :\; r_{a_1\omega_1+a_2\omega_2}\neq \emptyset}2\frac{\delta_1}{\delta_2}|a_1|^{2n-2}\sinh\left(\frac{\delta_2\delta}{(1+\alpha^2)^{1/2}}\right)e^{-\delta_2\frac{a_1+\alpha a_2}{1+\alpha^2}}<\infty\,.\\
    \end{split}
    \end{equation}
  
  On the other hand, on $r_{c}$ we simply have that due to the periodicity of $\wp$, the factor $\wp(\eta)$ is bounded and hence the integral over $r_c$ is also finite. We then conclude that 
  \begin{equation}
    \int_{r_{\infty}}\left|e^{-\eta/\epsilon}\wp(\eta)\mathrm{d}\eta\right|= \int_{r_{c}}\left|e^{-\eta/\epsilon}\wp(\eta)\mathrm{d}\eta\right|+\int_{r_{p}}\left|e^{-\eta/\epsilon}\wp(\eta)\mathrm{d}\eta\right|<\infty 
\end{equation}

    The argument for the convergence of 
    \begin{equation}
        \int_{r_{\infty}}\left|e^{-\eta/\epsilon}\zeta(\eta)\mathrm{d}\eta\right|
    \end{equation}
    is similar. The only difference is that now $\zeta$ is not periodic, so we must use the corresponding identity in Lemma \ref{lemma2}. This shows that for $\eta(t)\in r_{\omega}$ we have
\begin{equation}
        \zeta(\eta)=\zeta(\eta-\omega)+2(a_1\eta_1+a_2\eta_2)=\frac{1}{\eta-\omega}+\text{Reg}(\eta-\omega) + 2(a_1\eta_1+a_2\eta_2)\,,
    \end{equation}
    where Reg as before is a holomorphic function (independent of $\omega$) on a disc of radius $\delta$ centered at $0$, 
    so that 
    \begin{equation}
        |\zeta(\eta)|\leq \frac{D_1}{|\eta -\omega|}+D_2(|a_1|+|a_2|)<\frac{D_1 (1+\alpha^2)^{1/2}\cdot |a_1|^{n-1}}{C}+D_2(|a_1|+|a_2|)
    \end{equation}
    for some constants $C, D_1,D_2$ independent of $\omega=a_1\omega_1+a_2\omega_2$. The argument for the convergence over $r_p$ follows as before. For the convergence over $r_c$ we again use the quasi-periodicity of $\zeta$ from Lemma \ref{lemma2} as before to show that as we go to $\infty$ along $r_c$ we have 
    \begin{equation}
        |\zeta(\eta)|=\mathcal{O}(|\eta|)\,.
    \end{equation}
   
    Finally, to show 
    \begin{equation}
        \int_{r_{\infty}}\left|e^{-\eta/\epsilon}\rho_i(\eta)\mathrm{d}\eta\right|<\infty
    \end{equation}
    we use that 
    \begin{equation}
        \int_{r_{\infty}}\left|e^{-\eta/\epsilon}\zeta(\eta)\mathrm{d}\eta\right|<\infty
    \end{equation}
    together with Lemma \ref{lemma2} and a simple modification of the argument from before. 
\end{proof}
\subsection{Proof of the Borel summability}\label{borelsummabilityproofsec}

Given the previous lemmas, we now prove the Borel summability of $K_i(\epsilon \b \omega_1,\omega_2)$ and $F(\epsilon\b \omega_1,\omega_2)$.
\begin{prop}\label{borelsumproof} Fix $(\omega_1,\omega_2)\in \cR$ and 
    consider a non-Stokes ray $r$ such that $\mu(r)<\infty$ with respect to $\Lambda(\omega_1,\omega_2)$. Then for $\epsilon \in \mathbb{H}_r$ the following integrals are absolutely convergent
    \begin{equation}
        K_{i}^{r}(\epsilon\b\omega_1,\omega_2)=\int_{r} e^{-\eta/\epsilon} k_i(\eta\b\omega_1,\omega_2)\mathrm{d}\eta, \quad F^r(\epsilon\b\omega_1,\omega_2)=\int_{r} e^{-\eta/\epsilon} f(\eta\b\omega_1,\omega_2)\mathrm{d}\eta\,,
    \end{equation}
    and depend holomorphically on $\epsilon$.
    In particular, the formal series $K_i(\epsilon\b\omega_1,\omega_2)$ and $F(\epsilon\b\omega_1,\omega_2)$ are Borel summable along $r$. Additionally, we have the alternate expressions
    \begin{equation}\label{borelalt}
        \begin{split}
        K_{i}^{r}(\epsilon\b\omega_1,\omega_2)&=\int_r \mathrm{Li}_1(e^{- \eta/\epsilon})\mathscr{k}_i(\eta\b\omega_1,\omega_2)\mathrm{d}\eta\,,\\
        F^r(\epsilon\b\omega_1,\omega_2)&=\int_r \mathrm{Li}_2(e^{- \eta/\epsilon})\mathscr{f}(\eta\b\omega_1,\omega_2)\mathrm{d}\eta\,.
        \end{split}
    \end{equation}
\end{prop}
\begin{proof}
Using that along $r$ we have $|e^{- \eta/\epsilon}|<1$ and using that 
\begin{equation}
    |\mathrm{Li}_1(z)|=|\log(1-z)|<-\log(1-|z|), \quad |z|<1,
\end{equation}
to show the absolute convergence of the first expression of \eqref{borelalt} it is enough to show the convergence of the integral
    \begin{equation}\label{eq:absRH}
        -\int_r \log(1-|e^{- \eta/\epsilon}|)|\mathscr{k}_i(\eta)||\mathrm{d}\eta|\,.
    \end{equation}
    Since near $\eta=0$ we have
    \begin{equation}
        \mathscr{k}_i(\eta)=\mathcal{O}(\eta)
    \end{equation}
    the integral in \eqref{eq:absRH} has no issue near $\eta=0$. On the other hand, as $\eta \to \infty$ along $r$ we have
    \begin{equation}
        -\log(1-|e^{- \eta/\epsilon}|)\sim |e^{- \eta/\epsilon}|\,,
    \end{equation}
    so by Lemma \ref{lemma3} we have that \eqref{eq:absRH} is finite. On the other hand, by Fubini-Tonneli and changing variables we have
    \begin{equation}\label{altcomp}
    \begin{split}
        -\int_r \log(1-|e^{- \eta/\epsilon}|)&|\mathscr{k}_i(\eta)||\mathrm{d}\eta|=\int_r \sum_{m\geq 1}\frac{|e^{- m \eta/\epsilon}|}{ m}|\mathscr{k}_i(\eta)||\mathrm{d}\eta|=\sum_{m\geq 1}\int_r \frac{|e^{-m \eta/\epsilon}|}{m}|\mathscr{k}_i(\eta)||\mathrm{d}\eta|\\
        &=\sum_{m\geq 1}\int_r\frac{|e^{- \eta/\epsilon}|}{ m^2}|\mathscr{k}_i(\eta/m)||\mathrm{d}\eta|=\int_r|e^{- \eta/\epsilon}|\sum_{m\geq 1}\frac{1}{ m^2}|\mathscr{k}_i(\eta/m)||\mathrm{d}\eta|\,.
    \end{split}
    \end{equation}
    In the above, we have used that along the ray $|e^{-\eta/\epsilon}|<1$, so that the series expansion of $-\log(1-z)$ is valid along $r$. Since the first integral is finite and $k_i(\eta\b\omega_1,\omega_2)$ is given by \eqref{curlyk2}, it follows that the Borel sum $K_{i}^{r}(\epsilon\b\omega_1,\omega_2)$ is absolutely integrable. By applying Fubini-Tonelli to the expressions without absolute values, we also get the alternate identity in \eqref{borelalt}. 

    The argument for $F^r$ follows similarly. Using that along $r$ we have $|e^{- \eta/\epsilon}|<1$ and 
    \begin{equation}
        |\mathrm{Li_2}(z)|\leq \mathrm{Li_2}(|z|), \quad |z|<1
    \end{equation}
    to show the absolute convergence of the second expression in \eqref{borelalt} it is enough to consider 
    \begin{equation}\label{eq:abstau}
    \int_r  \mathrm{Li}_2(|e^{- \eta/\epsilon}|)|\mathscr{f}(\eta)||\mathrm{d}\eta|\,.
    \end{equation}
    Since $\mathrm{Li}_2(1)<\infty$ and the modified functions $\widetilde{\zeta}$ and $\widetilde{\wp}$ are finite at $\eta=0$, the integrand does not have any issues at $\eta=0$. Similar to the previous case, as $\eta \to \infty$ along $r$ we have
    \begin{equation}
        \mathrm{Li}_2(|e^{- \eta/\epsilon}|)\sim |e^{- \eta/\epsilon}| 
    \end{equation}
    so by Lemma \ref{lemma3} we find that \eqref{eq:abstau} is finite. By Fubini-Tonneli and performing a change of variables as in \eqref{altcomp}, we find that
    \begin{equation}
        \begin{split}
            \int_r  \mathrm{Li}_2(|e^{- \eta/\epsilon}|)|\mathscr{f}(\eta)||\mathrm{d}\eta|
            &=\sum_{m\geq 1}\int_r \frac{|e^{-m \eta/\epsilon}|}{m^2}|\mathscr{f}(\eta)||\mathrm{d}\eta|=\sum_{m\geq 1}\int_r \frac{|e^{- \eta/\epsilon}|}{m^3}|\mathscr{f}(\eta/m)||\mathrm{d}\eta|\\
            &=\int_r|e^{- \eta/\epsilon}|\sum_{m\geq 1}\frac{1}{m^3}|\mathscr{f}(\eta/m)||\mathrm{d}\eta|\,.\\
        \end{split}
    \end{equation}
    As before, for the series expansion of $\mathrm{Li}_2(z)$ we have used that along $r$ we have $|e^{- \eta/\epsilon}|<1$.
    Since the first integral is finite, and $f$ is given by \eqref{curlyf}, it follows that the Borel sum $F^r(\epsilon,\omega_1,\omega_2)$ is absolutely integrable. By applying Fubini-Tonelli to the corresponding expressions without absolute values, we obtain the alternate identity for $F^r$ in \eqref{borelalt}.

    Finally, we show holomorphic dependence in $\epsilon \in \mathbb{H}_r$ for $K^r_i$, with an identical argument for $F^r$. Consider any contour $\partial \Delta \subset \mathbb{H}_r$. Then we clearly have
    \begin{equation}
        \int_{\partial \Delta}\left(\int_{r}|e^{-\eta/\epsilon}k_i(\eta \b \omega_1,\omega_2)||\mathrm{d}\eta|\right)|\mathrm{d}\epsilon|<\infty\,.
    \end{equation}
    By applying Fubini-Tonelli we can interchange the order of integration, and we find
    \begin{equation}
    \begin{split}
    \int_{\partial \Delta} K_i^r(\epsilon \b \omega_1,\omega_2)\mathrm{d}\epsilon&=
        \int_{\partial \Delta}\left(\int_{r}e^{-\eta/\epsilon}k_i(\eta \b \omega_1,\omega_2)\mathrm{d}\eta\right)\mathrm{d}\epsilon\\
        &=\int_{r}\left(\int_{\partial \Delta}e^{-\eta/\epsilon}\mathrm{d}\epsilon\right)k_i(\eta \b \omega_1,\omega_2)\mathrm{d}\eta=0\,.
    \end{split}
    \end{equation}
    Hence, by Morera's theorem it follows that $K_i^r(\epsilon\b \omega_1,\omega_2)$ is holomorphic in $\epsilon \in \mathbb{H}_r$.
\end{proof}

\subsection{Integrals along detour paths}\label{detourpathproof}
When $\mu(\alpha)=\infty$, we can still define something meaningful. The idea is as follows:
\begin{itemize}
    \item Give a non-Stokes ray $r$ with $\mu(r)=\infty$ with respect to the lattice $\Lambda(\omega_1,\omega_2)$, let $r(\delta)$ be the detour path defined in Section \ref{Borelssummary} for $\delta$ small enough (see figure \ref{detourpathfig}). 
    \item We then define the following expressions for $\epsilon \in \mathbb{H}_r$
    \begin{equation}\label{borelalt2}
        \begin{split}
        K_{i}^{r(\delta)}(\epsilon\b\omega_1,\omega_2)&=\int_{r(\delta)} \mathrm{Li}_1(e^{- \eta/\epsilon})\mathscr{k}_i(\eta\b\omega_1,\omega_2)\mathrm{d}\eta\,,\\
        F^{r(\delta)}(\epsilon\b\omega_1,\omega_2)&=\int_{r(\delta)} \mathrm{Li}_2(e^{- \eta/\epsilon})\mathscr{f}(\eta\b\omega_1,\omega_2)\mathrm{d}\eta\,.
        \end{split}
    \end{equation}
    \item We will then show the above expressions are independent of $\delta$ for $\delta$ small enough, and they coincide with $K_{i}^{r}$ and $F^r$ when $\mu(r)<\infty$. 
\end{itemize}

\begin{prop}\label{weakRHsol}
   Let $r$ be a  non-Stokes ray with respect to $\Lambda(\omega_1,\omega_2)$. Then there is $D>0$ such that for all $0<\delta <D$ the integrals $F^{r(\delta)}$ and $K_{i}^{r(\delta)}$ are absolutely convergent  for $\epsilon \in \bH_r$. These integrals  depend holomorphically on $\epsilon\in \mathbb{H}_r$, and are independent of the choice of such $\delta$. Moreover when $\mu(r)<\infty$ we have $F^{r(\delta)}=F^{r}$ and  $K_{i}^{r(\delta)}=K_{i}^{r}$\,.

\end{prop}
\begin{proof}
    We take $D>0$ such that the discs of radius $0<\delta<D$ and centered at $\Lambda(\omega_1,\omega_2)$ do not intersect each other. Notice that given any parametrization $\eta(t)$ of the corresponding detour path $r(\delta)$ and $\epsilon \in \mathbb{H}_r$, we have $\text{Re}(\eta(t)/\epsilon)>0$ for $t$ sufficiently big, so we still have exponential decay as $t\to \infty$.

    On the other hand, the proof of the absolute convergence follows a simpler argument than the one used in Lemma \ref{lemma3} and Proposition \ref{borelsumproof}. Indeed, one first needs a version of Lemma \ref{lemma3} for the detour paths $r(\delta)$. As in Lemma \ref{lemma3} we can focus on a segment $r_{\infty}(\delta)$ given by $\eta(t)$ for $t\in [K,\infty)$ and $K>0$ sufficiently big, and furthemore divide $r_{\infty}(\delta)$ into two sets
    \begin{equation}
        r_{\infty}(\delta)=r_p\cup r_c
    \end{equation}
    where $r_c$ is exactly as in Lemma \ref{lemma3}, and $r_p$ is now made of the arcs of the detour path, belonging to circles of radius $\delta$ centered at the poles. The argument of the absolute convergence over $r_c$ is exactly the same as in Lemma \ref{lemma3}, while the estimates for $r_p$ are easier, since we are now always a bounded distance from the poles. For example, when dealing with $\wp(\eta)$, using the periodicity of $\wp(\eta)$ we simply have a uniform bound for $|\wp(\eta)|$ along $r_p$, while for $\zeta(\eta)$ and $\rho_i(\eta)$ we use again Lemma \ref{lemma2}. One can then apply the same argument of Proposition \ref{borelsumproof} to show that the integrals in \eqref{borelalt2} are absolutely integrable.\footnote{Note that in the case of rays $r$ with $\mu(r)=\infty$ we do not know that the expressions \eqref{borelalt} are absolutely integrable, and hence we cannot directly relate \eqref{borelalt2} to Borel sums by trying to show that $K_{i}^{r(\delta)}$ and $F^{r(\delta)}$ coincide with \eqref{borelalt}. This is because the argument of Proposition \ref{borelsumproof} requires the absolute convergence of \eqref{borelalt} to show that $\eqref{borelalt}$ match the corresponding Borel sums.} 

    Now let $\delta$ be small enough and $0<\delta'<\delta$. Let $r_{n}(\delta)$ (resp. $r_{n}(\delta')$) be the segment of the detour path $r(\delta)$ (resp. $r(\delta')$) from $0$ to some point between the $n$-th arc and the $(n+1)$-th arc. We pick the endpoint to be the same for $r_{n}(\delta)$ and $r_{n}(\delta')$ for all $n>0$. Similarly, we denote by $r_n$ the segment of $r$ from $0$ to the endpoint of $r_{n}(\delta)$. By a trivial argument with contour integrals using the fact that the integrands have poles at $\Lambda^*(\omega_1,\omega_2)$ 
    \begin{equation}\label{parteq}
        K_{i}^{r_n(\delta)}=K_{i}^{r_n(\delta')}=K_{i}^{r_n}, \quad F^{r_n(\delta)}=F^{r_n(\delta')}=F^{r_n}\,, \quad \text{for all $n>0$}\,.
    \end{equation}
    When $r$ is an arbitrary non-Stokes ray, we can use the existence of $K_{i}^{r(\delta)}$ and $F^{r(\delta)}$ for all small enough $\delta$ and \eqref{parteq} to take a limit $n\to \infty$ and obtain $K_{i}^{r(\delta)}=K_{i}^{r(\delta')}$, $F^{r(\delta)}=F^{r(\delta')}$\footnote{Note that this does not show that $K_{i}^{r(\delta)}=K_{i}^{r}$, $F^{r(\delta)}=F^{r}$, since for a general non-Stokes ray $r$ we do not know that $F^r$ and $K_{i}^{r}$ exist, and the existence of a limit along a sequence does not guarantee the existence of the limit.}. 
    Furthermore, when $r$ is a non-Stokes ray such that $\mu(r)<\infty$, we can use the existence of $K_{i}^{r}$ and $F^{r}$ and \eqref{parteq} to obtain $K_{i}^{r(\delta)}=K_{i}^{r}$, $F^{r(\delta)}=F^{r}$.

    Finally, the holomorphicity in $\epsilon \in \mathbb{H}_r$ follows by a the same kind of argument as in Proposition \ref{borelsumproof}
\end{proof}
\subsection{Stokes jumps}\label{stokesjumpssec}

Finally, we discuss how the previous integrals along different paths relate to each other. 
\begin{prop}
    Let $r_1$ and $r_2$ be two non-Stokes rays ordered in clockwise order, and assume that $\mathbb{H}_{r_1}\cap \mathbb{H}_{r_2}\neq \emptyset$. Furthermore, let $\Delta(r_1,r_2)$ be the sector determined by $r_1$ and $r_2$. Then for $\epsilon \in \mathbb{H}_{r_1}\cap \mathbb{H}_{r_2}$ and small enough $\delta$ we have 
    \begin{equation}
    \begin{split}
        K_{i}^{r_2(\delta)}(\epsilon \b \omega_1,\omega_2)-K_{i}^{r_1(\delta)}(\epsilon \b \omega_1,\omega_2)&=2\pi \mathrm{i}\sum_{\omega\in \Delta(r_1,r_2)\cap \Lambda^*(\omega_1,\omega_2)}a_i\cdot \log(1-e^{-\omega/\epsilon})\\
         F^{r_2(\delta)}(\epsilon,\omega_1,\omega_2)-F^{r_1(\delta)}(\epsilon,\omega_1,\omega_2)&=2\pi \mathrm{i}\sum_{\omega \in \Delta(r_1,r_2)\cap \Lambda^*(\omega_1,\omega_2)}\frac{\partial}{\partial \epsilon}\left(\epsilon \mathrm{Li}_2(e^{-\omega/\epsilon})\right)\,,
    \end{split}
    \end{equation}
    where $\omega=a_1\omega_1+a_2\omega_2$.
\end{prop}
\begin{proof}
    Consider a sequence $C_n$ with $n>0$ of discs centered at $0$ and of radius $R_n$ with $R_n \to \infty$ as $n\to \infty$. We denote by $A_n$ the arc of $C_n$ contained in $\Delta(r_1,r_2)$ and assume that $A_n$ does not intersect $\Lambda^*(\omega_1,\omega_2)$ for all $n$. We orient $A_n$ counter-clockwise. Furthermore, consider $\delta>0$ small enough such that the discs of radius $\delta>0$ centered at the points of $\Lambda(\omega_1,\omega_2)$ do not intersect. We consider detour arcs $A_n(\delta)$, defined similarly to the detour rays $r(\delta)$ by taking a detour along the circles of radius $\delta$ centered at the points of $\Lambda(\omega_1,\omega_2)$ through the shortest length arc. Furthermore, we denote by $\Delta_n(r_1,r_2)$ the region determined by $r_1$, $r_2$ and $A_n(\delta)$, and by $r_{n,1}(\delta)$ and $r_{n,2}(\delta)$ the segments of $r_1(\delta)$ and $r_2(\delta)$ from $0$ to the intersection points with $A_n(\delta)$. By using that $\mathscr{k}_i(\eta)$ has a simple pole at $\omega=a_1\omega_1+a_2\omega_2\in \Lambda^*(\omega_1,\omega_2)$ with residue $-a_i$ we obtain using \eqref{borelalt2} that
 \begin{equation}\label{Stokescont1}
        K_{i}^{r_{n,2}(\delta)}-K_{i}^{r_{n,1}(\delta)}+K_{i}^{A_n(\delta)}=2\pi \mathrm{i}\sum_{\omega\in \Delta_n(r^1,r^2)\cap \Lambda^*}a_i\cdot \log(1-e^{-\omega/\epsilon})
    \end{equation}
    where we have used that $\mathrm{Li}_1(z)=-\log(1-z)$. Similarly, using that $\widetilde{\wp}=-\mathrm{d}\widetilde{\zeta}/\mathrm{d}\eta$ and the definition of $\mathscr{f}$, we can use integration by parts on $F^{r(\delta)}$ to write
    \begin{equation}
    \begin{split}
        F^{r(\delta)}&=\int_{r(\delta)}\left(2\mathrm{Li}_2(e^{-\eta/\epsilon})-\frac{\mathrm{d}}{\mathrm{d}\eta}(\eta \mathrm{Li}_2(e^{-\eta/\epsilon}))\right)\widetilde{\zeta}(\eta,\omega_1,\omega_2)\mathrm{d}\eta\,\\
        &=\int_{r(\delta)}\left(\mathrm{Li}_2(e^{-\eta/\epsilon})+\frac{\eta}{\epsilon}\mathrm{Li}_1(e^{-\eta/\epsilon})\right)\widetilde{\zeta}(\eta,\omega_1,\omega_2)\mathrm{d}\eta\,\\
    \end{split}
    \end{equation}
    where we have used that the boundary terms of the integration by parts vanish. Using that $\widetilde{\zeta}$ has a simple pole at $\omega\in \Lambda^*(\omega_1,\omega_2)$ with residue $1$ then shows that 
    \begin{equation}\label{stokescont2}
    \begin{split}
        F^{r_{n,2}(\delta)}-F^{r_{n,1}(\delta)}+F^{A_n(\delta)}&=2\pi \mathrm{i}\sum_{\omega\in \Delta_n(r^1,r^2)\cap \Lambda^*}\left(\mathrm{Li}_2(e^{-\omega/\epsilon})+\frac{\omega}{\epsilon}\mathrm{Li}_1(e^{-\omega/\epsilon})\right)\\
        &=2\pi \mathrm{i}\sum_{\omega\in \Delta_n(r^1,r^2)\cap \Lambda^*}\partial_{\epsilon}\left(\epsilon \mathrm{Li}_2(e^{-\omega/\epsilon})\right)\\
    \end{split}
    \end{equation}
    Now note that since $\epsilon \in \mathbb{H}_{r^1}\cap \mathbb{H}_{r^2}$ the function $e^{-\eta/\epsilon}$ along $A_n(\delta)$ is exponentially suppressed as $n\to \infty$. By using a similar argument to Lemma \ref{lemma3}, Proposition \ref{borelsumproof} and Proposition \ref{weakRHsol} one then finds that
    \begin{equation}
        \lim_{n\to \infty}F^{A_n(\delta)}=\lim_{n\to \infty}K_{i}^{A_n(\delta)}=0\,
    \end{equation}
    and hence the result follows. 
\end{proof}
\appendix
\section{Free energy in fibre classes}\label{appendix}

Let $\pi:X\to B$ be an elliptic CY threefold satisfying the assumptions \ref{CYass}, and consider the GW generating function in fibre classes of $\pi:X \to B$ and for genus $g\geq 2$
\begin{equation}\label{GWfiber}
    F_{\text{GW}}(\lambda\b Q)=\sum_{g\geq 2}F_g(Q)\lambda^{2g-2}\,, \quad F_g(Q)=\sum_{n=0}^{\infty}\text{GW}(g,nF)Q^d\,.
\end{equation}
Here $\text{GW}(g,nF)$ denotes the GW invariant of the class $n F$ at genus $g$, and $F$ is the fiber class of $\pi:X \to B$.

In this Section we show that one can write $F_{\text{GW}}$ as in \eqref{free_energy}. The expression \eqref{FE1} below is the same as the one written in \cite[Section B.3]{OP}. We nevertheless include a more detailed computation for completeness.  

\begin{prop}
    Assuming the GW/DT correspondence holds for $X$, we can write $F_{\text{GW}}(\lambda \b Q)$ from \eqref{GWfiber} as
    \begin{equation}\label{FE1}
        F_{\text{GW}}(\lambda \b Q)=e(X)\sum_{g\geq2}\frac{(-1)^gB_{2g}}{4g}C_{2g-2}(Q)\lambda^{2g-2}\,,
    \end{equation}
    where $C_{2g-2}(Q)$ is the analytic function in $Q$ for $|Q|<1$ given by
    \begin{equation}
        C_{2g-2}(Q)=-\frac{B_{2g-2}}{(2g-2)\cdot (2g-2)!}+\frac{2}{(2g-2)!}\sum_{k,n\geq 1}k^{2g-3}Q^{kn}\,.
    \end{equation}
    Furthermore, setting $Q=e^{2\pi \mathrm{i}\tau}$ for $\text{Im}(\tau)>0$ we have
    \begin{equation}\label{FE3}
        F_{\text{GW}}(\lambda\b \tau)=-e(X)\cdot \sum_{g\geq 2} \frac{B_{2g} \, G_{2g-2}(\tau)}{4g(2g-2)} \left(\frac{\lambda}{2\pi}\right)^{2g-2}.
    \end{equation}
\end{prop}
\begin{proof}
    On the one hand, for $g\geq 2$ and $n=0$ we have the universal contribution of constant maps on a CY threefold $X$  \cite[Theorem 4]{FR}
    \begin{equation}\label{constmapcont}
        \text{GW}(g,0)=-e(X)\frac{(-1)^gB_{2g}B_{2g-2}}{4g(2g-2)(2g-2)!}\,\,.
    \end{equation}

    On the other hand, consider the Gopakumar-Vafa form of the Gromov-Witten generating function 
    \begin{equation}
        \sum_{g\geq 0, \; \beta>0}\text{GW}(g,\beta)Q^{\beta}\lambda^{2g-2}=\sum_{g\geq 0, \; \beta>0, \; k>0}\frac{\text{GV}(g,\beta)}{k}\left(2\sin\left(\frac{k\lambda}{2}\right)\right)^{2g-2}Q^{k\beta}\,,
    \end{equation}
    where $\text{GV}(g,\beta)$ denotes the Gopakumar-Vafa invariants of the class $\beta$ at genus $g$. Assuming the DT/GW correspondence it is shown in  \cite[Section 6]{T} that for $\beta=nF$, $n>0$, we have
    \begin{equation}
      \text{GV}(0,nF)=-e(X), \qquad \text{GV}(1,nF)=e(B), \qquad \text{GV}(g,nF)=0 \quad \text{for} \quad g\geq 2\,.  
    \end{equation}
    Hence, restricting to the sum over fiber classes (and denoting $Q^{n\cdot F}=Q^n$ to simplify notation) we find 
    \begin{equation}
    \begin{split}
        \sum_{g\geq 0, \; n>0}\text{GW}(g,nF)Q^{n}\lambda^{2g-2}&=-e(X)\sum_{k,n>0}\frac{1}{k}\left(2\sin\left(\frac{k\lambda}{2}\right)\right)^{-2}Q^{kn}+e(B)\sum_{k,n>0}\frac{Q^{kn}}{k}\,,\\
        &=e(X)\sum_{k,n>0}\frac{e^{\mathrm{\mathrm{i}k\lambda}}}{k(e^{\mathrm{i}k\lambda}-1)^2}Q^{kn}-e(B)\sum_{n> 0}\log(1-Q^n)\,.
    \end{split}
    \end{equation}
    Using the generating function of Bernoulli numbers $B_{n}$ and the fact that $B_{n}=0$ for odd $n>1$,  one easily finds that
    \begin{equation}
        \frac{e^{\mathrm{i}k\lambda}}{(e^{\mathrm{\mathrm{i}k\lambda}}-1)^2}=\sum_{g=0}^{\infty}\frac{(2g-1)(-1)^gB_{2g}(k\lambda)^{2g-2}}{(2g)!}\,\,,
    \end{equation}
    so when considering only terms with $g\geq 2$ we find that
    \begin{equation}
    \begin{split}
        \sum_{g\geq 2, \; n>0}\text{GW}(g,nF)Q^{n}\lambda^{2g-2}&=e(X)\sum_{g\geq 2}\frac{(-1)^gB_{2g}}{(2g)\cdot (2g-2)!}\left(\sum_{k,n>0}k^{2g-3}Q^{k\cdot n}\right)\lambda^{2g-2}\,\\
        &=e(X)\sum_{g\geq 2}\frac{(-1)^gB_{2g}}{4g}\left(C_{2g-2}(Q)+\frac{B_{2g-2}}{(2g-2)\cdot (2g-2)!}\right)\lambda^{2g-2}\,.
    \end{split}
    \end{equation}
 The expression \eqref{FE1} then follows by adding the constant map contribution \eqref{constmapcont}.

Finally, to show \eqref{FE3} note that the Eisenstein series $G_{2g-2}(\tau)$ has the following expansion for $g\geq 2$, which is a slight rewriting of its  Fourier series\footnote{In the Fourier expansion of $G_{2g-2}(\tau)$ a sum of the form $\sum_{m>0}\sigma_{2g-3}(m)Q^m$ appears, where $\sigma_{2g-3}(m)=\sum_{d | m}d^{2g-3}$ and $Q=e^{2\pi \mathrm{i}\tau}$. We simply use the fact that $\sum_{m>0}\sigma_{2g-3}(m)Q^{m}=\sum_{k,n>0}k^{2g-3}Q^{k\cdot n}$.}

\begin{equation}
    G_{2g-2}(\tau)=2\zeta(2g-2)\left(1-\frac{(2\pi)^{2g-2}(-1)^g}{(2g-3)!\zeta(2g-2)}\sum_{k,n>0}k^{2g-3}Q^{kn}\right), \quad Q=e^{2\pi \mathrm{i}\tau}\,.
\end{equation}
In the above $\zeta(s)$ denotes the Riemann $\zeta$-function and not the Weierstrass $\zeta$-function that is used in the rest of the paper. 
From
\begin{equation}
    \zeta(2g-2)=\frac{(-1)^g(2\pi)^{2g-2}B_{2g-2}}{2(2g-2)!}\,, \quad g\geq 2\,,
\end{equation}
it then follows that
\begin{equation}\label{altEis}
        C_{2g-2}(Q)=-\frac{(-1)^g}{(2g-2)(2\pi)^{2g-2}}G_{2g-2}(\tau)\,\,.
\end{equation}
Hence \eqref{FE3} follows from \eqref{altEis} and \eqref{FE1}.
\end{proof}

\section{Lemma on irrationality measure}\label{app_irrat}

Recall the definition of irrationality measure $\mu(\alpha)$ from Section \ref{irrat}. Here we prove that if $\alpha\in \bR\setminus\bQ$ then
    \begin{equation} 
   \mu\bigg(\frac{a\alpha+b}{c\alpha+d}\bigg)=\mu(\alpha)\text{ for all } \mat{a}{b}{c}{d}\in \GL_2(\bZ).
   \end{equation}
It it is enough to show that the result holds for the generators
\begin{equation}
    T=\mat{1}{1}{0}{1}, \quad S=\mat{0}{-1}{1}{0}, \quad P=\mat{1}{0}{0}{-1}\,.
\end{equation}
of $\mathrm{GL}_2(\mathbb{Z})$. This is obvious for the transformations $\alpha\mapsto \alpha+1$ and $\alpha\mapsto -\alpha$  corresponding to $T$ and $P$. Thus it remains to prove that $\mu(\alpha)=\mu(1/\alpha)$. By the invariance of $\mu(\alpha)$ under $\alpha\mapsto -\alpha$  we can assume that $\alpha>0$. 

Suppose that $0< r<\mu(\alpha)$. By definition of $\mu(\alpha)$, this implies that there are infinitely many $p,q\in \mathbb{Z}$ with $q>0$ such that 
\begin{equation}\label{in1}
    |\alpha-p/q|<1/q^r\,.
\end{equation}
For a given $q>0$ there can only be finitely many $p$ satisfying \eqref{in1}. So there must be a sequence $p_n,q_n\in \mathbb{Z}$ with $q_n>0$ satisfying \eqref{in1} such that $q_n\to \infty$ as $n\to \infty$. Then
\begin{equation}\label{lim}
    \frac{p_n}{q_n}\to \alpha\,.
\end{equation} and hence $p_n\to \infty$ also. Passing to a subsequence we can assume that all $p_n>0$. It follows from \eqref{in1} and \eqref{lim} that for some constant $C>0$
\begin{equation}
    \left|\frac{1}{\alpha}-\frac{q_n}{p_n}\right|<\frac{1}{q_n^{r-1}p_n\alpha}=\frac{p_n^{r-1}}{q_n^{r-1}p_n^r\alpha}<\frac{C}{p_n^r}\,.
\end{equation}

If $0<r'<r$ then using the fact that $p_n\to \infty$ we can assume, after possibly passing to another subsequence, that $p_n^{r-r'}>C$ for all $n\in \mathbb{N}$, and hence that
\begin{equation}
     \left|\frac{1}{\alpha}-\frac{q_n}{p_n}\right|<\frac{1}{p_n^{r'}}\,.
\end{equation}
This implies that $r' \leq \mu(1/\alpha)$. Since this holds for all $0<r'<r$, it follows that $r\leq \mu(1/\alpha)$. But $0<r<\mu(\alpha)$ was chosen  arbitrarily so we conclude that $\mu(\alpha)\leq \mu(1/\alpha)$. Repeating the argument interchanging $\alpha$ and $1/\alpha$ gives $\mu(\alpha)=\mu(1/\alpha)$.


\end{document}